 \documentclass[preprint,12pt,3p]{elsarticle}

\usepackage{amssymb}
\usepackage{amsthm}
\usepackage{amsmath}

\usepackage{psfig}
\usepackage{hyperref}

\usepackage{xcolor,enumerate,array}
\usepackage[color,all]{xy}

\input xy
\xyoption{all}

\theoremstyle{plain}
\newtheorem{theorem}{Theorem}

\newtheorem{proposition}[theorem]{Proposition}
\newtheorem{lemma}[theorem]{Lemma}
\newtheorem{example}[theorem]{Example}

\theoremstyle{definition}
\newtheorem{definition}[theorem]{Definition}
\newtheorem{note}[theorem]{Note}

\numberwithin{equation}{section}
\numberwithin{theorem}{section}

\numberwithin{equation}{section}
\numberwithin{theorem}{section}
\RequirePackage[normalem]{ulem} 
\RequirePackage{color}\definecolor{RED}{rgb}{1,0,0}\definecolor{BLUE}{rgb}{0,0,1} 

\biboptions{sort&compress}

\journal{J. Differential Equations}

\begin{document}

\begin{frontmatter}



\title{A Lie systems approach to the Riccati hierarchy \\ and partial differential equations}


\author{J. de Lucas}
\address{Department of Mathematical Methods in Physics, University of Warsaw, \\ ul. Pasteura 5, 02-093, Warszawa, Poland.}
\author{A.M. Grundland}
\address{Centre de Recherches Math\'ematiques, Universit\'e de Montr\'eal,\\
C.P. 6128, Succ. Centre-Ville, Montr\'eal (QC) H3C 3J7, Canada \\
Department of Mathematics and Computer Science, Universit\'e du Qu\'ebec \`a Trois-Rivi\`eres, \\
Trois-Rivi\`eres, CP 500, G9A 5H7, Qu\'ebec, Canada}

\begin{abstract}It is proved that the members of the Riccati hierarchy, the so-called {\it  Riccati chain equations}, can be considered as particular cases of projective Riccati equations, which greatly simplifies the study of the Riccati hierarchy. This also allows us to characterize Riccati chain equations geometrically in terms of the projective vector fields of a flat Riemannian metric and to easily derive their associated superposition rules. Next, we establish necessary and sufficient conditions under which it is possible to map second-order Riccati chain equations into conformal Riccati equations through a local  diffeomorphism. This fact can be used to determine superposition rules for particular higher-order Riccati chain equations which depend on fewer particular solutions than in the general case. Therefore, we analyze the properties of Euclidean, hyperbolic and projective vector fields on the plane in detail. Finally, the use of contact transformations enables us to apply the derived results to the study of certain integrable partial differential equations, such as the Kaup--Kupershmidt and Sawada--Kotera equations. 
\end{abstract}


\end{frontmatter}

\section{Introduction}

The main objective of this work is to show that the members of the Riccati hierarchy can be understood as projective Riccati equations and to use this fact to analyze their geometric properties, superposition rules and related partial differential equations (PDEs). This approach allows us to obtain results that would be difficult to obtain by analyzing the Riccati hierarchy straightforwardly.

The Riccati hierarchy is of primary importance in the field of integrable systems (see e.g. \cite{GL99}). The first element of this hierarchy is, up to a change of the independent variable, the Riccati equation, namely
\begin{equation}\label{GenRic}
\frac{{\rm d}u}{{\rm d}x}=a_0(x)+a_1(x)u+a_2(x)u^2,\qquad u,x\in \mathbb{R},
\end{equation}
where $a_0(x),a_1(x),a_2(x)$ are arbitrary real-valued functions \cite{BLW91}. Riccati equations
frequently appear in physics, mathematics, control theory, astronomy and many other subjects (see \cite{Re72} and references therein). Mathematically, a Riccati equation can be understood as the differential equation describing the integral curves of a non-autonomous vector field taking values in a Lie algebra of vector fields on $\mathbb{R}$. This algebra is isomorphic to $\mathfrak{sl}(2,\mathbb{R})$, as noted by Lie \cite{1880} and Vessiot \cite{Ve83}. A modern framework for these facts has been developed by Cari\~nena, Grabowski and Marmo \cite{CGM07,Dissertations}. 

Despite its apparent simplicity, there is no general method for obtaining the general solution of a generic non-autonomous Riccati equation \cite{In44}. Nevertheless, the general solution  of a Riccati equation can be brought into the form
$$
u(x)=\frac{u_{(1)}(x)(u_{(3)}(x)-u_{(2)}(x))+ku_{(2)}(x)(u_{(1)}(x)-u_{(3)}(x))}{u_{(3)}(x)-u_{(2)}(x)+k(u_{(1)}(x)-u_{(3)}(x))},
$$
where $u_{(1)}(x),u_{(2)}(x),u_{(3)}(x)$ are different particular solutions and $k$ is an arbitrary real constant. This property of the Riccati equation is called a {\it superposition principle} \cite{Dissertations,CGM07,PW}. 

More generally, a non-autonomous first-order system of ordinary differential equations whose general solution can be described as an autonomous function of a generic set of particular solutions and some constants, a so-called {\it superposition rule}, is called a {\it Lie system} \cite{Dissertations,CIMM15,CGM07,PW}. Sophus Lie proved that each Lie system is related to a finite-dimensional Lie algebra of vector fields, called a {\it Vessiot--Guldberg Lie algebra}, which describes many of  its properties \cite{LS,PW}.

The second-order Riccati chain equation is a generalization of the Painlev\'e-Ince equation  \cite{In44,KL09,Ve95} and takes the form
$$
\frac{{\rm d}^2u}{{\rm d}x^2}+(\alpha_2(x)+3cu)\frac{{\rm d}u}{{\rm d}x}+c^2u^3+c\alpha_2(x)u^2+\alpha_1(x)u+\alpha_0(x)=0,
$$
where $c\in\mathbb{R}^*:=\mathbb{R}\backslash\{0\}$ and $\alpha_0(x),\alpha_1(x),\alpha_2(x)$ are arbitrary $x$-dependent functions.   
This differential equation appears, for instance, in the study of B\"acklund transformations \cite{GL99} and it  has recently been studied in \cite{CL10SecOrd2,CRS05,GL13}.
Second-order Riccati chain equations become Lie systems when written as a first-order system by adding a new variable, $v:={\rm d}u/{\rm d}x$ \cite{CL10SecOrd2}, which can be used to obtain a superposition rule for this particular case \cite{CL10SecOrd2}. 


Each member of the Riccati hierarchy is called an {\it $s$-order Riccati chain equation}, where $s\in \mathbb{N}$. It was proved in \cite{GL13} that each $s$-order Riccati chain equation is a Lie system related to a Vessiot--Guldberg Lie algebra isomorphic to $\mathfrak{sl}(s+1,\mathbb{R})$. 

As a first new result, Theorem \ref{MT1} provides a family of diffeomorphisms $\{\phi_{c,s}\}_{(c,s)\in \mathbb{R}^*\!\times\mathbb{N}}$ mapping each member of the Riccati hierarchy, considered as a first-order system in the standard way, into a projective Riccati equation. The diffeomorphisms under consideration are globally defined, which allows us to establish that Riccati chain equations can be studied as particular types of 
projective Riccati ones. This allows us to recover, as a particular example, the results of \cite{CGR15} concerning first- and second-order Riccati equations. 

The diffeomorphisms $\{\phi_{c,s}\}_{(c,s)\in \mathbb{R}^*\!\times{ \mathbb{N}}}$ possess several advantages which are absent in the previous literature on the Riccati hierarchy. First, they transform the complicated form of the $s$-order Riccati chain equations into simpler projective Riccati equations (see e.g. Table \ref{table1} and equations (\ref{ProRicc})). Second, the diffeomorphisms allow us to prove that Riccati chain equations are, essentially, the  Lie systems determined by a finite-dimensional Vessiot--Guldberg Lie algebra of projective vector fields relative to a flat Riemannian metric. This easily allows us to determine when a system of differential equations can be mapped into a projective Riccati equation. The diffeomorphisms $\{\phi_{c,s}\}_{(c,s)\in \mathbb{R}^*\!\times{ \mathbb{N}}}$ can be understood as changes of variables mapping the flat Riemannian metrics associated with the Riccati chain equations into diagonal forms. This extends in a very simple way  the relations of 
\cite{CGR15} between the very lowest members of the Riccati hierarchy and projective vector fields to the whole hierarchy.
Third, since superposition rules for projective Riccati equations are known \cite{AW80}, the family $\{\phi_{c,s}\}_{(c,s)\in \mathbb{R}^*\!\times{ \mathbb{N}}}$ enables us to obtain a superposition rule for all $s$-order Riccati chain equations as first-order systems. This is a much more powerful approach than the one provided in \cite{CL10SecOrd2}, where only second-order Riccati chain equations were considered.  

Our characterization of the Riccati hierarchy in terms of projective Riccati equations constitutes a new way of characterizing and studying second-order Riccati chain equations that can be mapped through a diffeomorphism $\phi:{\rm T}\mathbb{R}\rightarrow \mathbb{R}^2$ into conformal Euclidean and/or hyperbolic Riccati equations, respectively. This characterization is described in Theorems \ref{MT2} and \ref{MT3}.

As a consequence of the technique described above, we obtain new results concerning the structure of conformal and projective Lie algebras of vector fields on the plane. In particular,  Table \ref{table5} summarizes all new results on projective and conformal Lie algebras of vector fields on $\mathbb{R}^2$ given in Propositions \ref{ConLie} to \ref{HypPro}. Table \ref{table5} also includes all the relation inclusions of these Lie algebras, which can be obtained after a straightforward but lengthy calculation. We derive the so-called invariant distributions for all finite-dimensional Lie algebras of vector fields on the plane (see the last column of Table \ref{table3}). 

It is also proved that second-order Riccati chain equations which are not autonomous cannot be described through a Lie system related to a Vessiot--Guldberg Lie algebra of Hamiltonian vector fields with respect to a symplectic structure, namely a Lie--Hamilton system \cite{CLS13,Gul,Ve83}. For the so-called {\it second-order affine Riccati chain equations}, the necessary and sufficient conditions which ensure that these equations can be described through Lie--Hamilton systems are determined. 

Next, we show that certain B\"acklund transformations for partial differential equations can be studied through projective Riccati equations and we prove that second-order Riccati chain equations can be mapped through a contact transformation into equations of the Gambier family. This makes it possible to study Gambier equations G25 and G27 and the related partial differential equations via Lie systems. As a particular instance, we apply our methods to the Sawada--Kotera  and Kaup-Kupershmidt equations. Finally, the relation between Gambier equation G25 and  the Sturm-Liouville problem is analyzed. 

The plan of the paper is as follows. In Section 2 we describe the fundamental geometric properties of Lie systems and related notions. The basic properties of the Riccati hierarchy are discussed in detail in Section 3. We prove in Section 4 that every member of the Riccati hierarchy can be mapped onto a projective equation through an autonomous diffeomorphism.  Section 5 is devoted to proving that Riccati chain equations are Lie systems with a Vessiot--Guldberg Lie algebra of projective vector fields relative to a flat Riemannian metric. The results found in Section 4 are used in Section 6 to obtain superposition rules for the whole Riccati hierarchy. Section 7 is concerned with proving certain new results on the structure and relations of Lie algebras of conformal and projective vector fields on the plane. The results of Section 7 are used in Section 8 to classify all second-order Riccati chain equations related to Vessiot--Guldberg Lie algebras of conformal vector fields.
 We then show how to use our results in the study of partial differential equations such as the Kaup--Kupershmidt and Sawada--Kotera equations in Section 9. A contact transformation is used to relate certain differential equations to members of the Riccati hierarchy in Section 10. A last application of Gambier equations related to Lie systems is analyzed in Section 11. The last section summarizes the obtained results and contains some suggestions regarding possible further developments.
  
\section{Fundamentals}\label{LSLS}
The methodological approach used in this work is based on the study of non-autonomous systems of first-order differential equations by means of vector fields along projections. If not otherwise stated, all structures are assumed to be smooth. To simplify the notation and to avoid unnecessary technical problems, diffeomorphisms between structures are considered to be local and defined at generic points unless explicitly expressed the contrary.

A {\it vector field} on $N$ along a projection $\pi_N:P\rightarrow N$ is a map $X:p\in P\mapsto X_p\in TN$ 
for which $\tau_N\circ X=\pi_N$, where $\tau_N:TN\rightarrow N$ is the tangent bundle projection onto $N$. If we assume that $P=\mathbb{R}\times N$, and we call $x$ the canonical variable on $\mathbb{R}$, then $X$ is called a {\it non-autonomous} or {\it $x$-dependent vector field}. An $x$-dependent vector field amounts to
a family of vector fields $\{X_x\}_{x\in\mathbb{R}}$ with $X_x:u\in N\mapsto
X(x,u)\in TN$ for all $x\in\mathbb{R}$ and vice versa \cite{Dissertations}.  We assume hereafter  that $X$  represents a non-autonomous vector field.

We call {\it integral curves} of $X$ the integral curves 
$\gamma:\mathbb{R}\mapsto \mathbb{R}\times N$ of the {\it suspension} of $X$,
i.e. the vector field $X(x,u)+\partial/\partial x$ on $\mathbb{R}\times N$ \cite{FM}. Every
integral curve $\gamma$ admits a reparametrization $
x=x(t)$ such that $\gamma(x)=( x, u(x))$ and
$$
\frac{{\rm d}(\pi_N \circ \gamma)}{{\rm d}x}(x)=(X\circ \gamma)( x).
$$
This system is referred to as the {\it associated system} of $X$. Conversely,
every non-autonomous system of first-order ordinary differential equations in normal form describes the
integral curves of a unique non-autonomous vector field. This establishes a
bijection between non-autonomous vector fields and systems of first-order ordinary
differential equations in normal form, which
 justifies the use of $X$ to denote both a non-autonomous vector field and its
associated system.

\begin{definition}  The  {\it irreducible Lie algebra} of an $x$-dependent vector
field $X$ on $N$ is the smallest (in the sense of inclusion) real Lie algebra, 
$V^X$, containing the vector fields $\{X_x\}_{x\in\mathbb{R}}$.  
\end{definition}

\begin{definition} Given a finite-dimensional Lie algebra of vector fields $V$ on $N$, its {\it associated distribution} is 
the generalized distribution 
$
\mathcal{D}^V_p:=\{X_p:X\in V\}\subset T_pN,\,\, \forall p\in N.
$
A Lie algebra $V$ of vector fields on $\mathbb{R}^2$ is called {\it primitive} when its elements do not leave any one-dimensional distribution on $\mathbb{R}^2$ invariant (when acted on through 
Lie brackets). Otherwise, we say that $V$ is {\it imprimitive}. If $V$ admits one or more invariant distributions, we say that $V$ is {\it mono-imprimitive} or {\it multi-imprimitive}, respectively.
\end{definition}

For instance, the {\it conformal vector fields} relative to a pseudo-Riemannian metric $g$ on $N$, i.e. the vector fields $X$ on $N$ satisfying  $\mathcal{L}_Xg=f_Xg$ for a certain function $f_X\in C^\infty(N)$ called the {\it potential function} of $X$, form a Lie algebra of vector fields. Given the Lie algebras of vector fields P$_7$ and I$_{11}$ (see Table \ref{table3}), it is known that P$_7\simeq\mathfrak{so}(3,1)$ is a maximal finite-dimensional Lie algebra of conformal polynomial vector fields on $\mathbb{R}^2$ relative to a Euclidean metric and I$_{11}\simeq \mathfrak{so}(2,2)$ is a maximal Lie algebra of conformal polynomial vector fields on $\mathbb{R}^2$ relative a hyperbolic metric (cf. \cite{BL00,GKP92}).

Although Lie classified Vessiot--Guldberg Lie algebras on $\mathbb{R}^2$, the result was not clarified until the work of Olver, Artemio and Kamran \cite{GKP92}. Table \ref{table2}, the so-called GKO (Gonz\'alez-Kamran-Olver) classification, details all Vessiot--Guldberg Lie algebras described in  \cite{GKP92}. To analyze Riccati chain equations, it is convenient to obtain the invariant distributions of all such Lie algebras. These invariant distributions can be obtained algorithmically after a long but straightforward calculation. Therefore, we detail in Table \ref{table2} the invariant distributions for Vessiot--Guldberg Lie algebras on $\mathbb{R}^2$ with no further details.

 Let us now turn to some fundamental notions appearing in the theory of Lie
systems. 

\begin{definition} A {\it superposition rule} depending on $m$ particular
solutions for a system $X$ on $N$ 
 is a function $\Phi:N^{m}\times N\rightarrow
N$, $u=\Phi(u_{(1)}, \ldots,u_{(m)};\lambda)$ such that the general
 solution $u(x)$ of $X$ can be brought into the form
  $u(x)=\Phi(u_{(1)}(x), \ldots,u_{(m)}(x);\lambda),$
where $u_{(1)}(x),\ldots,u_{(m)}(x)$ is any generic family of
particular solutions and $\lambda$ is an arbitrary  element of $N$. 
 \end{definition}

The conditions ensuring that a system $X$ possesses a superposition rule are
given by the {\it Lie--Scheffers Theorem} \cite[Theorem 44]{LS} (for a modern geometric description see \cite[Theorem 1]{CGM07} and \cite{Dissertations,OG00}). 

\begin{theorem} {\bf (Lie--Scheffers Theorem)} A system $X$ on $N$ admits a superposition rule if and only if $X={{\sum_{\alpha=1}^r}}b_\alpha(x)X_\alpha$ 
for a certain family $b_1(x),\ldots,b_r(x)$  of $x$-dependent functions and a
collection  $X_1,\ldots,X_r$ of vector fields on $N$  spanning 
an $r$-dimensional real Lie algebra.
\end{theorem}

Non-autonomous systems of first-order ordinary differential equations possessing a superposition rule
are called {\it Lie systems}. The Lie--Scheffers 
Theorem states that every Lie system $X$ is related to (at least) one
finite-dimensional real Lie algebra of vector fields 
$V$, a so-called {\it Vessiot--Guldberg Lie algebra}, satisfying
$\{X_x\}_{x\in\mathbb{R}}\subset V$. The irreducible Lie algebra of $X$ allows us to rewrite more intrinsically the Lie--Scheffers Theorem as follows
\cite{Dissertations}. 

\begin{theorem}\label{ALST}{\bf (Abbreviated Lie--Scheffers Theorem)} A
system $X$ admits a superposition rule if and only if $V^X$ is
finite-dimensional.
\end{theorem}

Finally, let us describe some types of Lie systems relevant to the present work. Strictly speaking (see \cite{AHW80}), projective Riccati equations on $\mathbb{R}^n$ take the form
\begin{equation*}
\frac{{\rm d}\xi}{{\rm d}x}=b_0(x)+A(x)\xi-\gamma(x)\xi +\langle \xi,b_2(x)\rangle \xi,\qquad \xi\in \mathbb{R}^n,
\end{equation*}
where $A(x)$ is an $n\times n$ matrix with real coefficients, $b_0(x),b_2(x)\in \mathbb{R}^n$, we assume that $\gamma(x)$ is an $x$-dependent scalar function such that ${\rm Tr} [A(x)]+\gamma(x)=0$, and $\langle\cdot,\cdot\rangle$ is the canonical Euclidean metric on $\mathbb{R}^n$. In general, projective Riccati equations can be rewritten as
$$
\frac{{\rm d}\xi}{{\rm d}x}=b_0(x)+[A(x)+p(x) {\rm Id}_n]\xi-(p(x)+\gamma(x))\xi+\langle \xi,b_2(x)\rangle \xi,\qquad 
$$
for $n\,p(x):=\gamma(x),$ which ensures that $A(x)+p(x){\rm Id}_n$ is a traceless matrix. Hence, a projective Riccati  equation can be written as a differential equation of the form
\begin{equation}\label{ProRicc}
\frac{{\rm d}\xi}{{\rm d}x}=b_0(x)+C(x)\xi+\langle \xi,b_2(x)\rangle \xi,\qquad \xi\in \mathbb{R}^n,
\end{equation}
where $C(x)$ is an $n\times n$ matrix with real entries, and vice versa. This allows us to simplify the expression of the projective Riccati equations. Each projective Riccati equation on $\mathbb{R}^n$ is a  Lie system associated with a Vessiot--Guldberg Lie algebra $V_n^{\rm Pr}\simeq \mathfrak{sl}(n+1,\mathbb{R})$ \cite{AHW80}.

Meanwhile, a conformal Riccati equation on the plane takes the form \cite{AW80}
$$
\frac{{\rm d}\xi}{{\rm d}x}=b_0(x)+A(x)\xi+\gamma(x)\xi +b_2(x)\langle \xi,\xi\rangle-2\langle \xi,b_2(x)\rangle \xi,\qquad \xi\in \mathbb{R}^n,
$$
where $\langle A(x)\xi_1,\xi_2\rangle+\langle \xi_1,A(x)\xi_2\rangle=0$ for every $\xi_1,\xi_2\in \mathbb{R}^n$, the function $\gamma(x)$ is  an arbitrary $x$-dependent scalar function and $\langle\cdot,\cdot\rangle$ is a non-degenerate metric of signature $(p,q)$ with $p+q=n$. Conformal Riccati equations are Lie systems related to a Vessiot--Guldberg Lie algebra $V^{(p,q)}$ of conformal vector fields relative to a flat metric $g$ of signature $(p,q)$ and therefore isomorphic to $\mathfrak{so}(p+1,q+1)$ \cite{AHW81}.  

This work is mainly concerned with two types of conformal Riccati equations. The first one is the conformal Riccati equation on $\mathbb{R}^2$ related to the hyperbolic metric $\langle \xi,\bar \xi\rangle=\xi_1\bar \xi_2+\xi_2\bar \xi_1$ with $\xi:=(\xi_1,\xi_2)^T,\bar\xi:=(\bar \xi_1,\bar \xi_2)^T\in\mathbb{R}^2$. The corresponding conformal Riccati equation takes the form
\begin{equation}
 \begin{array}{l}
\frac{{\rm d}\xi_1}{{\rm d}x}=b^u_0(x)+B_{uu}(x)\xi_1-2\xi^2_1b_2^v(x),\\
\frac{{\rm d}\xi_2}{{\rm d}x}=b^v_0(x)+B_{vv}(x)\xi_2-2\xi^2_2b_2^u(x).
 \end{array}\label{system}
\end{equation}
for arbitrary real functions $b^u_0(x),b^v_0(x),B_{uu}(x),B_{vv}(x),b_2^u(x),b_2^v(x)$. It can be seen that equations (\ref{system}) are related to  non-autonomous vector fields taking values in the Lie algebra I$_{11}$ given in Table \ref{table3}, which consists of conformal vector fields relative to the metric ${\rm d}\xi_1\otimes {\rm d}\xi_2+{\rm d}\xi_2\otimes {\rm d}\xi_1$. For the sake of brevity, we will hereafter call (\ref{system}) a {\it hyperbolic Riccati equation}.

Meanwhile, a conformal Riccati equation related to the Euclidean metric $\langle \xi,\bar \xi\rangle=\xi_1\bar \xi_1+\xi_2\bar \xi_2$ with $\xi:=(\xi_1,\xi_2)^T,\bar\xi:=(\bar \xi_1,\bar \xi_2)^T\in\mathbb{R}^2$ takes the form
\begin{equation}\label{EuRic}
\begin{aligned}
\frac{{\rm d}\xi_1}{{\rm d}x}&=b^u_0(x)+B_{uu}(x)\xi_1-B_{uv}(x)\xi_2+(\xi^2_2-\xi_2^1)b_2^u(x)-b_2^v(x)2\xi_1\xi_2,\\
\frac{{\rm d}\xi_2}{{\rm d}x}&=b^v_0(x)+B_{uv}(x)\xi_1+B_{uu}(x)\xi_2+(\xi^2_1-\xi_2^2)b_2^v(x)-b_2^u(x)2\xi_1\xi_2,
\end{aligned}
\end{equation}
for certain functions $b^u_0(x),b^u_0(x),B_{uu}(x),B_{uv}(x),b_2^u(x),b_2^v(x)$. 
This system of differential equations  is related to a non-autonomous vector field taking values in the Lie algebra P$_{7}$ of Table \ref{table3}, which is made of conformal vector fields relative to the metric ${\rm d}\xi_1\otimes {\rm d}\xi_1+{\rm d}\xi_2\otimes {\rm d}\xi_2$. For conciseness, we will refer to equations of the form (\ref{EuRic}) as {\it Euclidean Riccati equations}.

Since there exists no finite-dimensional Lie algebra of vector fields containing the Vessiot--Guldberg Lie algebra $V^{(p,q)}$ \cite{BL00,GKP92}, there exists no diffeomorphism $\phi:\mathbb{R}^n\rightarrow \mathbb{R}^n$ mapping all conformal Riccati equations into projective ones. Nevertheless, particular conformal Riccati equations, e.g. autonomous ones, can be mapped into cases of projective Riccati equations.

\section{Introduction to the Riccati hierarchy}\label{Appl}

An {\it $s$-order Riccati chain equation} \cite{GL99} is a differential equation of the form
\begin{equation}\label{sRicc}
L_c^su+\sum_{j=1}^s\alpha_j(x)L_c^{j-1}u+\alpha_0(x)=0,\qquad u,x\in\mathbb{R},\qquad c\in \mathbb{R}^*,\qquad s\in \mathbb{N},
\end{equation}
where $\alpha_0(x),\ldots,\alpha_s(x)$ are arbitrary $x$-dependent real functions, $L^s:=L\circ \cdots \circ L (s-{\rm times})$, $L_c^0u:=u$,  and $L_c$ is the differential operator on the real line given by
\begin{equation}\label{L}
L_c:=\frac{\rm d}{{\rm d}x}+c\,u,\qquad c\in\mathbb{R}.
\end{equation}
There exists a more general definition of $s$-order Riccati chain equations, but it is equivalent to ours  through a simple change of the independent variable \cite{CRS05}. For instance, the first element of the most general Riccati hierarchy is (\ref{GenRic}), while in our case the first element is given in Table \ref{table1} for $s=1$. A trivial $x$-dependent change of variables maps one into the other. In view of these remarks, we can restrict ourselves to (\ref{sRicc}). For the sake of completeness, we will also consider the hierarchy referred to as the {\it $s$-order affine Riccati chain equations}, which is  given by (\ref{sRicc}) for $c=0$.

Expressions (\ref{sRicc}) and (\ref{L}) show that each $s$-order affine Riccati chain equation is affine, which motivates the term. Otherwise,  (\ref{sRicc}) can be linearized through {\it the Cole--Hopf transformation} \cite{Ho50}
$$
u(x):=\frac{1}{c\Psi}\frac{{\rm d}\Psi}{{\rm d}x},\qquad \Psi:x\in\mathbb{R}\mapsto \Psi(x)\in\mathbb{R},\qquad c\in \mathbb{R}^*,
$$
giving rise to the $(s+1)$-order linear differential equation
$$
\begin{gathered}
\sum_{j=0}^s\alpha_j(x)\frac{{\rm d}^j\Psi}{{\rm d}x^j}+\frac{{\rm d}^{s+1}\Psi}{{\rm d}x^{s+1}}=0,
\end{gathered}
$$
with ${\rm d}^0\Psi/{\rm d}x^0:=\Psi$.
In Table \ref{table1} we find the first members of the Riccati hierarchy. Let us analyze them to illustrate some of their properties.

\begin{table}[t] {\footnotesize
 \noindent
\caption{{\small First elements of the Riccati hierarchy. The Lie algebra $\mathfrak{g}$ is isomorphic to the Vessiot--Guldberg Lie algebra associated with the Lie system obtained by writing the $s$-order Riccati chain equation as a first-order system in the standard way (see \cite{GL13} for details).}}
\label{table1}
\medskip
\noindent\hfill
 \begin{tabular}{ |m{1cm}| m{1.2cm}  |  p{10.5cm} |}
\hline
&  &\\[-1.9ex]
$s$&$\mathfrak{g}$ & Riccati chain equation\\[+1.0ex]
\hline
 &  &\\[-1.9ex]
$1$&$\mathfrak{sl}(2,\mathbb{R})$ & $\frac{{\rm d}u}{{\rm d}x}+cu^2+\alpha_1(x)u+\alpha_0(x)=0$ \\[+1.0ex]
$2$&$\mathfrak{sl}(3)$ & $\frac{{\rm d}^2u}{{\rm d}x^2}+(\alpha_2(x)+3cu)\frac{{\rm d}u}{{\rm d}x}+c^2u^3+c\alpha_2(x)u^2+\alpha_1(x)u+\alpha_0(x)=0$\\[+1.0ex]
$3$&$\mathfrak{sl}(4,\mathbb{R})$ &$\frac{{\rm d}^3u}{{\rm d}x^3}+(\alpha_3(x)+4cu)\frac{{\rm d}^2u}{{\rm d}x^2}+3c\left(\frac{{\rm d}u}{{\rm d}x}\right)^2+[6c^2u^2+3c\alpha_3(x)u+\alpha_2(x)]\frac{{\rm d}u}{{\rm d}x}+c^3u^4$\\
&&$+c^2\alpha_3(x)u^3+c\alpha_2(x)u^2+\alpha_1(x)u+\alpha_0(x)=0$\\[+1.0ex]

\hline
 \end{tabular}
\hfill}
\end{table}

The first differential equation in Table \ref{table1} is a particular type of {Riccati equation}. It is associated with the $x$-dependent vector field
$$
X^{\rm 1R}=-\alpha_0(x)X^{\rm 1R}_0-\alpha_1(x)X^{\rm 1R}_1-cX^{\rm 1R}_2,
$$
where $X^{\rm 1R}_0:=\partial/\partial u$, $X^{\rm 1R}_1:=u\partial/\partial u$ and $X^{\rm 1R}_2:=u^2\partial/\partial u$. These vector fields satisfy the commutation relations
$$
[X^{\rm 1R}_0,X^{\rm 1R}_1]=X^{\rm 1R}_0,\qquad [X^{\rm 1R}_0,X^{\rm 1R}_2]=2X^{\rm 1R}_1,\qquad [X^{\rm 1R}_1,X^{\rm 1R}_2]=X^{\rm 1R}_2.
$$
Hence, $V_{\rm 1}^{\rm RC}:=\langle X^{\rm 1R}_0,X^{\rm 1R}_1,X^{\rm 1R}_2\rangle$ becomes a Lie algebra of vector fields isomorphic to $\mathfrak{sl}(2,\mathbb{R})$. It is worth noting that every Lie algebra of vector fields on $\mathbb{R}$ is locally diffeomorphic around a generic point to a Lie subalgebra of $V_1^{\rm RC}$ \cite{GKP92,1880}.

The second member of the Riccati hierarchy is the second-order Riccati chain equation, which can be written as a first-order system by adding a new variable $v:={\rm d}u/{\rm d}x$. This gives rise to the system on  ${\rm T}\mathbb{R}$ given by
\begin{equation}\label{RicSOr}
\left\{\begin{aligned}
\frac{{\rm d}u}{{\rm d}x}&=v,\\
\dfrac{{\rm d}v}{{\rm d}x}&=-3cuv-c^2u^3-\alpha_0(x)-\alpha_1(x)u-\alpha_2(x)\left(cu^2+v\right). 
\end{aligned}\right.
\end{equation}
The $x$-independent change of variables $\bar u:=cu$ and $\bar v:=cv$, with $c\in\mathbb{R}^*$, maps this system into the Lie system studied in \cite{CL10SecOrd2,GL13}, which is related to a Vessiot--Guldberg Lie algebra isomorphic to $\mathfrak{sl}(3)$. As a consequence, (\ref{RicSOr}) is a Lie system associated with a Vessiot--Guldberg Lie algebra $V^{\rm RC}_{2}\simeq\mathfrak{sl}(3)$. There exists no finite-dimensional Lie algebra of vector fields on $\mathbb{R}^2$ containing $V^{\rm RC}_{2}$ different from $V^{\rm RC}_{2}$ and every Lie algebra of vector fields on $\mathbb{R}^2$ isomorphic to $\mathfrak{sl}(3)$ is locally diffeomorphic to $V^{\rm RC}_{2}$ (cf. \cite{GKP92,1880}).

The system (\ref{RicSOr}) is related to the $x$-dependent vector field
$$
X^{\rm RC}_2=X^{\rm 2R}_3-\alpha_0(x)X^{\rm 2R}_0-\alpha_1(x)X^{\rm 2R}_1-\alpha_2(x)X^{\rm 2R}_2,
$$
where the vector fields $X^{\rm 2R}_0,X^{\rm 2R}_1,X^{\rm 2R}_2,X^{\rm 2R}_3$ belong to the following family of vector fields on ${\rm T}\mathbb{R}$.
\begin{equation}\label{VF}
\begin{aligned}
X^{\rm 2R}_0&:=\frac{\partial}{\partial  v},\,\,&X^{\rm 2R}_1&:=u\frac{\partial}{\partial v},\,\,\\
X^{\rm 2R}_2&:=(cu^2+v)\frac{\partial}{\partial v}\,\,&X^{\rm 2R}_3&:=v\frac{\partial}{\partial u}-(3cuv+c^2u^3)\frac{\partial}{\partial v},\,\, \\
X^{\rm 2R}_4&:=cu^2\frac{\partial}{\partial u}+cu(v-cu^2)\frac{\partial}{\partial v},\,\, &X^{\rm 2R}_5&:= c u(v+cu^2)\frac{\partial}{\partial  u}+ c (v^2-c^2u^4)\frac{\partial}{\partial v},\\
X^{\rm 2R}_6&:=u\frac{\partial}{\partial u}+2v\frac{\partial}{\partial v},\,\,
&X^{\rm 2R}_7&:=\frac{\partial}{\partial u}.
\end{aligned}
\end{equation}
Previous vector fields span a Lie algebra isomorphic to $\mathfrak{sl}(3)$. Indeed, we find
{\small\begin{equation}\label{Rel}
\begin{array}{llll}
\left[X^{\rm 2R}_0,X^{\rm 2R}_1\right]=0, &[X^{\rm 2R}_0,X^{\rm 2R}_2]=X^{\rm 2R}_0,&\left[X^{\rm 2R}_0,X^{\rm 2R}_3\right]=X^{\rm 2R}_7-3cX^{\rm 2R}_1,\\
\left[X^{\rm 2R}_0,X^{\rm 2R}_4\right]=cX^{\rm 2R}_1,&
[X^{\rm 2R}_0,X^{\rm 2R}_5]=cX^{\rm 2R}_6, &[X^{\rm 2R}_0,X^{\rm 2R}_6]=2X^{\rm 2R}_0,  \\
\left[X^{\rm 2R}_0,X^{\rm 2R}_7\right]=0,            &[X^{\rm 2R}_1,X^{\rm 2R}_2]=X^{\rm 2R}_1,&
\left[X^{\rm 2R}_1,X^{\rm 2R}_3\right]=X^{\rm 2R}_6-3X^{\rm 2R}_2,\\
\left[X^{\rm 2R}_1,X^{\rm 2R}_4\right]=0,&[X^{\rm 2R}_1,X^{\rm 2R}_5]=X^{\rm 2R}_4,&[X^{\rm 2R}_1,X^{\rm 2R}_6]=X^{\rm 2R}_1,\\
\left[X^{\rm 2R}_1,X^{\rm 2R}_7\right]=-X^{\rm 2R}_0,&[X^{\rm 2R}_2,X^{\rm 2R}_3]=X^{\rm 2R}_3+X^{\rm 2R}_4,         &[X^{\rm 2R}_2,X^{\rm 2R}_4]=0,     \\
\left[X^{\rm 2R}_2,X^{\rm 2R}_5\right]=X^{\rm 2R}_5,&
\left[X^{\rm 2R}_2,X^{\rm 2R}_6\right]=0,   &[X^{\rm 2R}_2,X^{\rm 2R}_7]=-2cX^{\rm 2R}_1,\\
\left[X^{\rm 2R}_3,X^{\rm 2R}_4\right]=X^{\rm 2R}_5,&[X^{\rm 2R}_3,X^{\rm 2R}_5]=0,&
\left[X^{\rm 2R}_3,X^{\rm 2R}_6\right]=-X^{\rm 2R}_3,\\
\left[X^{\rm 2R}_3,X^{\rm 2R}_7\right]=3cX^{\rm 2R}_2,&\left[X^{\rm 2R}_4,X^{\rm 2R}_5\right]=0,&\left[X^{\rm 2R}_4,X^{\rm 2R}_6\right]=-X^{\rm 2R}_4,\\
\left[X^{\rm 2R}_4,X^{\rm 2R}_7\right]=3c X^{\rm 2R}_2-2cX^{\rm 2R}_6,&[X^{\rm 2R}_5,X^{\rm 2R}_6]=-2X^{\rm 2R}_5,&\left[X^{\rm 2R}_5,X^{\rm 2R}_7\right]=-cX^{\rm 2R}_3-3cX^{\rm 2R}_4,\\
&\left[X^{\rm 2R}_6,X^{\rm 2R}_7\right]=-X^{\rm 2R}_7.\\
\end{array}
\end{equation}} 

Each $s$-order Riccati chain equation, when written as a first-order system on the $(s-1)$-order tangent bundle ${\rm T}^{s-1}\mathbb{R}$ by considering the derivatives of the dependent variables as new coordinates \cite{LR88}, is related to a Vessiot--Guldberg Lie algebra of vector fields isomorphic to $\mathfrak{sl}(s+1,\mathbb{R})$ \cite{GL13}. As in previous cases, we can obtain a Vessiot--Guldberg Lie algebra of vector fields for each $s$-order Riccati chain equation. Nevertheless, the expressions for the vector fields are nonlinear and become increasingly complicated. Meanwhile, the expression (\ref{RicSOr}) shows that the second-order affine Riccati chain equation becomes an affine system associated with a Vessiot--Guldberg Lie algebra isomorphic to ${\rm Aff}(\mathbb{R}^2)$, i.e. the Lie algebra of affine transformations on $\mathbb{R}^2$.

\section{The Riccati hierarchy and projective Riccati equations}

Let us proceed to prove one of the key new results of the paper: the $s$-order Riccati chain equation, written as a first-order system, is globally diffeomorphic to a projective Riccati equation on $\mathbb{R}^s$. Hence, Riccati chain equations are nothing but projective Riccati equations. This new approach simplifies the study of the Riccati hierarchy allowing its study through new techniques. In particular, we study in detail second- and third-order Riccati chain equations.

Consider the diffeomorphism 
\begin{equation}\label{phi2}
\phi_{2,c}:(u,v)\in {\rm T}\mathbb{R}\mapsto (y_1:=u,y_2:=cu^2+v)^T\in\mathbb{R}^2.
\end{equation}
 This map transforms the first-order system (\ref{RicSOr}), related to second-order Riccati chain equations, into 
\begin{equation}
\left\{\begin{aligned}
\frac{{\rm d}y_1}{{\rm d}x}&=y_2-cy_1^2,\\
\dfrac{{\rm d}y_2}{{\rm d}x}&=-\alpha_0(x)-\alpha_1(x)y_1-\alpha_2(x)y_2-cy_1y_2. \label{RicSOr2}
\end{aligned}\right.
\end{equation}
If we set $\xi:=(y_1,y_2)^T$, $b_0(x):=(0,-\alpha_0(x))^T$ and $b_2(x):=(-c,0)^T$, the latter system becomes a projective Riccati equation (\ref{ProRicc}) for $n=2$ and 
$$
C(x):=\left(\begin{array}{cc}0&1\\-\alpha_1(x)&-\alpha_2(x)\end{array}\right),
$$
In other words, the diffeomorphism $\phi_{2,c}$ maps second-order Riccati chain equations in first-order form into projective Riccati equations on $\mathbb{R}^2$.

A natural question arises: is there a general procedure to map any $s$-order Riccati chain equation, considered as a first-order system in the usual way, into a projective Riccati equation? 
The following theorem provides a global diffeomorphism mapping $s$-order Riccati chain equations onto projective Riccati equations on $\mathbb{R}^s$.

\begin{theorem}\label{MT1} An $s$-order, possibly affine, Riccati chain equation (\ref{sRicc}), when written as a first-order system on ${\rm T}^{s-1}\mathbb{R}$, can be mapped onto the projective Riccati equation (\ref{ProRicc}) on $\mathbb{R}^s$ 
with 
{\footnotesize
$$
b_0(x):=\left(\begin{array}{c}0\\0\\ \cdots\\0\\-\alpha_0(x)\end{array}\right),\quad
C(x):=\left(\begin{array}{ccccc}0&1&0&\ldots&0\\
0&0&1&\ldots&0\\\ldots&\ldots&\ldots&\ldots&\ldots\\0&0&0&\ldots&1\\-\alpha_1(x)&-\alpha_2(x)&-\alpha_3(x)&\ldots&-\alpha_n(x)\end{array}\right),\quad b_2(x):=\left(\begin{array}{c}-c\\0\\ \cdots\\0\\0\end{array}\right),
$$}via a global diffeomorphism $\phi_{s,c}:(u^{0)},\ldots,u^{s-1)})\in {\rm T}^{s-1}\mathbb{R}\mapsto (y_1,\ldots,y_s)^T\in \mathbb{R}^s$, with ${\rm T}^0\mathbb{R}:=\mathbb{R}$ and 
\begin{equation}\label{change}
y_{k}(x):=L_c^{k-1}u(x),\qquad k=1,\ldots,s.
\end{equation}
\end{theorem}
\begin{proof} Let us prove that (\ref{change}) gives rise to a global diffeomorphism $\phi_{s,c}:{\rm T}^{s-1}\mathbb{R}\rightarrow \mathbb{R}^s$. We have that
\begin{equation}\label{recurrence}
y_k(x)=L_c^{k-1}u(x)=u^{k-1)}(x)+F_{k,c}(u(x),u^{1)}(x),\ldots, u^{k-2)}(x)),\qquad k=1,\ldots,s,
\end{equation}
where $L_c^0u(x):=u(x)$ and $u^{k)}$ stands for the variable corresponding to the $k$-th derivative of $u$ in terms of $x$ and $u^{0)}:=u$. Here, $F_{k,c}:{\rm T}^{k-2}\mathbb{R}\rightarrow \mathbb{R}$ is such that $F_{1,c}=0$ for every $c\in \mathbb{R}$. Hence, 
$$
\frac{\partial y_k}{\partial u^{j-1)}}=0,\qquad s\geq j>k\geq 1,\qquad \qquad \frac{\partial y_k}{\partial u^{k-1)}}=1,\qquad k=1,\ldots,s.
$$
Thus, the Jacobian matrix $J_s$ corresponding to the transformation (\ref{change}) takes the form $J_s={\rm Id}_s+T_{s,c}$, where  ${\rm Id}_s$ is an $s\times s$ identity matrix and $T_{s,c}$ is an $s\times s$ lower triangular matrix with zeros in the main diagonal.
Hence, $\det J_s\neq 0$ and the inverse function theorem ensures that $\phi_{s,c}$ is locally invertible with a locally diffeomorphic inverse.


Let us prove that $\phi_{s,c}$ is surjective, i.e.  the algebraic equation $\phi_{s,c}(u^{0)},\ldots, u^{s-1)})=(y_1,\ldots, y_s)$ always has a solution for any $(y_1,\ldots,y_s)\in \mathbb{R}^s$. In view of the definition of $\phi_{s,c}$, this equation implies that $u^{0)}=y_1$. From (\ref{recurrence}), we have 
$$
u^{k-1)}=y_k-F_{k,c}(u,u^{1)},\ldots, u^{k-2)}),\qquad s\geq k>1,
$$ 
and every $u^{k-1)}$, with $s\geq k>1$, can be recursively and uniquely determined from the value of the $y_k$ and the lower derivatives to ensure that $\phi_{s,c}(u^{0)},\ldots, u^{s-1)})=(y_1,\ldots, y_s)$. Since $u^{0)}$ is established, we can obtain a unique solution of the algebraic equation $\phi_{s,c}(u^{0)},\ldots, u^{s-1)})=(y_1,\ldots, y_s)$, and therefore $\phi_{s,c}$ is surjective. As a consequence, 
$$\phi_{s,c}(u^{0)},\ldots, u^{s-1)})=\phi_{s,c}(\bar{u}^{0)},\ldots, \bar{u}^{s-1)})\quad \Rightarrow (u^{0)},\ldots, u^{s-1)})=(\bar u^{0)},\ldots, \bar u^{s-1)})$$
 and $\phi_{s,c}$ is also injective. Since each $\phi_{s,c}$ is a bijection with a locally differentiable inverse, it follows that  $\phi_{s,c}$ is a diffeomorphism.

Finally, we prove that the diffeomorphism given by (\ref{change}) maps $s$-order Riccati chain equations in first-order form onto projective Riccati equations. Substituting  (\ref{change}) in the definition of the $s$-order Riccati chain equation, i.e.
$$
L_c^su+\sum_{j=1}^s\alpha_j(x)(L_c^{j-1}u)+\alpha_0(x)=0,
$$
we obtain
$$
\frac{{\rm d}y_s}{{\rm d}x}=-cy_1y_{s}-\sum_{j=1}^s\alpha_j(x)y_{j}-\alpha_0(x).
$$
Meanwhile,
$$
y_k:=L^s_cy_{k-1}=\frac{{\rm d}y_{k-1}}{{\rm d}x}+cy_1y_{k-1}\Rightarrow \frac{{\rm d}y_{k-1}}{{\rm d}x}=y_k-cy_1y_{k-1}.
$$
Defining $\xi:=(y_1,\ldots,y_s)^T$, since $\langle \xi,b_2(x)\rangle=-cy_1$ and in view of the above expressions, the system (\ref{ProRicc}) is obtained.

\end{proof}

\begin{example}{\bf (Third-order Riccati chain equations)}
When written as a first-order system on the second-order tangent manifold ${\rm T}^2\mathbb{R}$, the third-order Riccati chain equation (see e.g. Table \ref{table1}) takes the form
{\small
\begin{equation}\label{ThirdOr}
\begin{aligned}
\frac{{\rm d}u}{{\rm d}x}&=v,\\
\frac{{\rm d}v}{{\rm d}x}&=a,\\
\dfrac{{\rm d}a}{{\rm d}x}&=-\alpha_3(x)(a+3cuv+c^2u^3)-\alpha_2(x)(cu^2+v)-\alpha_1(x)u-\alpha_0(x)-c(4ua+3v^2+c^2u^4+6cu^2v). 
\end{aligned}
\end{equation}}

The diffeomorphism induced by (\ref{change}) for $s=3$, i.e.
\begin{equation}\label{phi3}
\phi_{3,c}:(u,v,a)\in {\rm T^2}\mathbb{R}\mapsto (y_1:=u,y_2:=v+cu^2,y_3:=a+3cuv+c^2u^3)^T\in\mathbb{R}^3
\end{equation}
 transforms the above system into

\begin{equation}\label{ThirdOrConf}
\left\{\begin{aligned}
\frac{{\rm d}y_1}{{\rm d}x}&=y_2-cy_1^2,\\
\frac{{\rm d}y_2}{{\rm d}x}&=y_3-cy_1y_2,\\
\dfrac{{\rm d}y_3}{{\rm d}x}&=-\alpha_0(x)-\alpha_1(x)y_1-\alpha_2(x)y_2-\alpha_3(x)y_3-cy_3y_1. 
\end{aligned}\right.
\end{equation}

If we set $\xi:=(y_1,y_2,y_3)^T$, $b_0(x):=(0,0,-\alpha_0(x))^T$ and $b_2(x):=(-c,0,0)^T$, then (\ref{ThirdOr}) can be written intrinsically as a projective Riccati equation (\ref{ProRicc}) with
$$
C(x):=\left(\begin{array}{ccc}0&1&0\\0&0&1\\-\alpha_1(x)&-\alpha_2(x)&-\alpha_3(x)\end{array}\right).
$$
This is the system described in Theorem \ref{MT1}.
\end{example}
\section{Projective vector fields and Riccati chain equations} 
Let us show that the functions $\{\phi_{c,s}\}_{(c,s)\in\mathbb{R}\times \mathbb{N}}$ mapping Riccati chain equations into projective Riccati equations  allow us to prove that Riccati chain equations are first-order systems of differential equations taking values in the finite-dimensional projective Lie algebra of vector fields relative to a flat Riemannian metric. Therefore, the mappings $\{\phi_{c,s}\}_{(c,s)\in\mathbb{R}\times {\mathbb{N}}}$  can be understood as changes of variables mapping the flat Riemannian metrics into diagonal forms.

Given a pseudo-Riemannian manifold $(N,g)$, a {\it projective vector field} $Z$ on $N$ is a vector field satisfying the condition that there exists a one-form $\mu_Z$ on $N$ such that
\begin{equation}\label{sym1}
(\mathcal{L}_Z\nabla)(Z_1,Z_2)=\mu_Z(Z_1)Z_2+\mu_Z(Z_2)Z_1,\qquad \forall Z_1,Z_2 \in \mathfrak{X}(N),
\end{equation}
where $\mathfrak{X}(N)$ is the space of vector fields on $N$, the operator $\nabla$ is the covariant derivative induced by $g$ and  $\mathcal{L}_Z\nabla$ is the Lie derivative of $\nabla$ \cite{Ha04,Yo83}. The one-form $\mu_Z$ is called the {\it projective one-form} of $Z$ (relative to $g$). 
More conveniently for our purposes and assuming that our metric is flat, equation (\ref{sym1}) can be rewritten as \cite{Ha04}
\begin{equation}\label{sym}
\nabla_{Z_1}\nabla_{Z_2}Z-\nabla_{\nabla_{Z_1}{Z_2}}Z=\mu_Z({Z_1}){Z_2}+\mu_Z({Z_2}){Z_1}, \qquad \forall Z_1,Z_2\in \mathfrak{X}(N).
\end{equation}
The main property of projective vector fields is that their flows map geodesics of the metric $g$ into new geodesics (without necessarily preserving the affine parameter) \cite{RS02,Yo83}. Projective vector fields span a Lie algebra of vector fields.

We now show an interesting fact about Riccati equations and second-order Riccati chain equations.  The Vessiot--Guldberg Lie algebra of vector fields for Riccati equations, namely $\langle \partial_u,u\partial_u,u^2\partial_u\rangle$, consists of projective vector fields relative to the flat Riemannian metric $g_1:={\rm d}u\otimes {\rm d}u$.  We now turn to second-order Riccati chain equations in the first-order form (\ref{RicSOr}).
This system is related to the $x$-dependent vector field 
$
X^{\rm RC}_2=X^{\rm 2R}_3-\sum_{j=0}^2\alpha_j(x)X^{\rm 2R}_j,
$
where $X_\alpha^{\rm 2R}$ for $\alpha=0,1,2,3$ are given by (\ref{VF}).
Consider the Riemannian metric 
$$g_2:={\rm d}u\otimes {\rm d}u+{\rm d}(cu^2+v)\otimes{\rm d}(cu^2+v)=(1+4c^2u^2){\rm d}u\otimes {\rm d}u+{\rm d}v\otimes {\rm d}v+2cu({\rm d}u\otimes {\rm d}v+{\rm d}v\otimes {\rm d}u).
$$
A straightforward calculation shows that all Christoffel symbols for $g_2$ vanish, the only exception being $\Gamma^{v}_{uu}=2c$. Therefore, the Riemann tensor associated with $g_2$ also vanishes and $g_2$ becomes flat.  

Since (\ref{sym}) for $g_2$ is $C^\infty({\rm T}\mathbb{R})$-linear relative to $Z_1$ and $Z_2$, it is enough to check that it is satisfied  for a generator system of the $C^\infty({\rm T}\mathbb{R})$-module of vector fields on ${\rm T}\mathbb{R}$ to prove that it holds for any pair of vector fields on ${\rm T}\mathbb{R}$. Consider a generator system given by $\partial_1:=\partial_u,\partial_2:=\partial_v$. It is straightforward to verify that
 $\nabla_{\partial_\alpha}\nabla_{\partial_\beta } X^{\rm 2R}_i-\nabla_{\nabla_{\partial_\alpha} \partial_\beta}X^{\rm 2R}_i=0$ for $\alpha,\beta=1,2$ and $i=0,1,2$ and $\nabla_{\partial_\alpha}\nabla_{\partial_\beta } X^{\rm 2R}_3-\nabla_{\nabla_{\partial_\alpha} \partial_\beta}X^{\rm 2R}_3=-c{\rm d}u(\partial_\alpha)\partial_\beta-c{\rm d}u(\partial_\beta)\partial_\alpha$.  Hence, $X^{\rm 2R}_0,\ldots,X^{\rm 2R}_3$ are projective vector fields relative to $g_2$. In view of the commutation relations in (\ref{Rel}), these vector fields generate $V_2^{\rm RC}$. Hence, all elements of $V_2^{\rm RC}$ are projective vector fields relative to $g_2$. A natural question arises: is this only a property of first- and second-order Riccati chain equations or it is a general property of Riccati chain equations? The following theorem answers this question.

\begin{theorem}\label{ProjRicc} Every $s$-order Riccati chain equation admits a Vessiot--Guldberg Lie algebra $V_s^{\rm RC}$ of projective vector fields relative to the flat Riemannian metric
\begin{equation}\label{FC}
g^{\rm RC}_{s}:=\sum_{i=0}^{s-1}{\rm d}(L^iu)\otimes {\rm d}(L^iu).
\end{equation}
Additionally, $V_s^{\rm RC}\simeq \mathfrak{sl}(s+1,\mathbb{R})$ and $V_2^{\rm RC}$ is diffeomorphic to ${\rm P}_8$.
\end{theorem}

\begin{proof} As Riccati equations were shown to fulfill trivially the statement of the present theorem, we hereafter assume that $s>1$.
Let us sketch the outline of our proof. We first endow projective Riccati equations on $\mathbb{R}^{s}$ with a trivial flat Riemannian metric $g^s_P$ turning its Vessiot--Guldberg Lie algebra, $V^{\rm Pr}_s$, into projective vector fields. Afterwards, the transformation (\ref{change}) will allow us to map $s$-order Riccati chain equations onto Riccati projective equations on $\mathbb{R}^{s}$ while satisfying $\phi_{s,c}^*g^s_P=g^{\rm RC}_s$ and $\phi_{s,c*}V_s^{\rm RC}=V^{\rm Pr}_s$, where  $V^{\rm Pr}_s$ is the Vessiot--Guldberg Lie algebra for projective Riccati equations and $V_s^{\rm RC}$ is the Vessiot--Guldberg Lie algebra for $s$-order Riccati chain equations. Since the elements of $V^{\rm Pr}_s$ are projective relative to $g^s_P$, it follows that the elements of $V^{\rm RC}_s$ become projective relative to $g^{\rm RC}_s$ and  every $s$-order Riccati chain equation becomes, as a first-order system, a Lie system related to a $t$-dependent vector field taking values in the projective Lie algebra on ${\rm T}^{s-1}\mathbb{R}$ relative to $g^{\rm RC}_s$.

Consider the projective Riccati equation on $\mathbb{R}^{s}$. Let us solve (\ref{sym1}) for the flat Riemannian metric $g^s_P:=\sum_{i=1}^s{\rm d}y^i\otimes {\rm d}y^i$. Since (\ref{sym})  is $C^\infty({\rm T}^s\mathbb{R})$-linear with respect to the vector fields $Z_1,Z_2$, we can prove this expression by analyzing  it for $Z_1=\partial_i:=\partial/\partial y^i$ and $Z_2=\partial_j:=\partial/\partial y^j$ with $i,j=1,\ldots,s$. Christoffel symbols for $g_P^s$ are identically zero in the coordinate system $\{y^1,\ldots,y^s\}$. Hence, $\nabla_{Z_1}Z_2=D_{Z_1}Z_2=0$ where $D_{X_1}X_2$ is the directional derivative of $X_2$ relative to $X_1$. Moreover, (\ref{sym1}) becomes
$$
D_{Z_1}D_{Z_2}Z=\mu_Z(Z_1)Z_2+\mu_Z(Z_2)Z_1,
$$
 In our given system of coordinates, we can write $Z=\sum_{i=1}^sZ^i\partial_i$, $\mu_Z=\sum_{i=1}^s\mu^Z_idy^i$ and
\begin{equation}\label{ConEq}
\frac{\partial^2 Z^i}{\partial y^j\partial y^k}=\mu^Z_j\delta^i_k+\mu^Z_k\delta_j^i,\qquad i,j,k=1,\ldots,s.
\end{equation}
Hence, we obtain that \cite{Ud}
$$
\frac{\partial \mu^Z_j}{\partial y^l}\delta^i_k+\frac{\partial \mu^Z_k}{\partial y^l}\delta_j^i=\frac{\partial^3 Z^i}{\partial y^l\partial y^j\partial y^k}=\frac{\partial^3 Z^i}{\partial y^j\partial y^l\partial y^k}=\frac{\partial \mu^Z_l}{\partial y^j}\delta^i_k+\frac{\partial \mu^Z_k}{\partial y^j}\delta_l^i\Rightarrow \frac{\partial \mu^Z_j}{\partial y^l}=\frac{\partial \mu^Z_l}{\partial y^j},
$$
for all $l,j,k=1,\ldots,s.$ Moreover, we have
$$
0=\sum_{i=j=1}^s\left(\frac{\partial \mu^Z_j}{\partial y^l}\delta^i_k+\frac{\partial \mu^Z_k}{\partial y^l}\delta_j^i-\frac{\partial \mu^Z_l}{\partial y^j}\delta^i_k-\frac{\partial \mu^Z_k}{\partial y^j}\delta_l^i\right)=(s-1)\frac{\partial \mu^Z_k}{\partial y^l}=0, \qquad \forall k,l=1,\ldots,s.
$$
Since $s>1$,we get  $\mu^Z_k=c_k\in \mathbb{R}$ for $k=1,\ldots,s$. Integrating (\ref{ConEq}) we obtain that
\begin{equation}\label{MRE}
Z=\sum_{i=1}^s\left(a^i+\sum_{k=1}^sb^i_ky^k+y^i\sum_{k=1}^sc_{k}y^k\right)\frac{\partial}{\partial y^i},
\end{equation}
for arbitrary real numbers $a^i,b^i_k,c_k$ with $i,k=1,\ldots,s$.
This shows that projective Riccati equations admit a Vessiot--Guldberg Lie algebra of projective vector fields relative to a flat Riemannian metric spanned by the vector fields (\ref{MRE}). It is well known that these vector fields span a Lie algebra $V^{\rm Pr}_s\simeq\mathfrak{sl}(s+1,\mathbb{R})$ \cite{PW}.

The diffeomorphism $\phi_{s,c}:{\rm T}^{s-1}\mathbb{R}\rightarrow \mathbb{R}^{s}$ induces the flat Riemannian metric 
$
g^{\rm RC}_s:=\phi_{s,c}^*g^s_P=\sum_{i=1}^s({\rm d}L^iu)\otimes ({\rm d}L^iu)
$
on ${\rm T}^{s-1}\mathbb{R}$. The mapping $\phi_{s,c}$ maps $s$-order Riccati chain equations into projective Riccati equations. Hence, it also maps the Vessiot--Guldberg Lie algebra $V^{\rm RC}_s$ for  $s$-order Riccati chain equations into the Vessiot--Guldberg Lie algebra $V^{\rm Pr}_s$ for projective Riccati equations. As a consequence, the vector fields of $V^{\rm RC}_s$ become projective relative to $g^{\rm RC}_s$. 
\end{proof}

\begin{example}{\bf (Third-order Riccati chain equations)}
 Let us construct the flat Riemannian structure associated with third-order Riccati chain equations.
When written as a first-order system on the second-order tangent manifold ${\rm T}^2\mathbb{R}$, the third-order Riccati chain equation (\ref{ThirdOr}) is related to the $x$-dependent vector field $X^{\rm RC}_3:=X_4^{\rm 3R}-\sum_{j=0}^3\alpha_j(x)X^{\rm 3R}_j$ for
$$\begin{gathered}
X^{\rm 3R}_4:=v\frac{\partial}{\partial u}+a\frac{\partial}{\partial v}-c(4ua+3v^2+c^2u^4+6cu^2v)\frac{\partial}{\partial a},\quad X^{\rm 3R}_3:=(a+3cvu+c^2u^3)\frac{\partial}{\partial a}\\
 X^{\rm 3R}_2:=(cu^2+v)\frac{\partial}{\partial a},\quad X^{\rm 3R}_1:=u\frac{\partial}{\partial a},\quad X^{\rm 3R}_0:=\frac{\partial}{\partial a}.
\end{gathered}
$$
Recall that the diffeomorphism induced by (\ref{change}) for $s=3$ reads
\begin{equation}\label{phi3}
\phi_{3,c}:(u,v,a)\in {\rm T^2}\mathbb{R}\mapsto (y_1:=u,y_2:=v+cu^2,y_3:=a+3cuv+c^2u^3)\in\mathbb{R}^3.
\end{equation}

In accordance with our previous theorem, we define the metric $g^{\rm RC}_3=\sum_{j=1}^3 {\rm d}y^j\otimes {\rm d}y^j$, namely
\begin{multline}
g^{\rm RC}_3=(1+4c^2u^2+9(vc+c^2u^2)^2){\rm d}u\otimes {\rm d}u+(2cu+9c^2uv+9c^3u^3)({\rm d}u\otimes {\rm d}v+{\rm d}v\otimes {\rm d}u)\\
3cu({\rm d}a\otimes {\rm d}v+{\rm d}v\otimes{\rm d}a) +(1+9u^2c^2){\rm d}v\otimes {\rm d}v+{\rm d}a\otimes {\rm d}a+
(3cv+3c^2u^2)({\rm d}a\otimes {\rm d}u+{\rm d}u\otimes {\rm d}a).
\end{multline}
A simple calculation shows that $g^{\rm RC}_3$ is non-degenerate, and therefore a Riemannian metric. The nonvanishing Christoffel symbols for this metric read
$
\Gamma^v_{uu}=2c,\,\, \Gamma^a_{uv}=\Gamma^{a}_{vu}=3c
$ 
and its Riemann tensor vanishes.  Hence, $g^{\rm RC}_3$ is flat. Using these results, we can easily prove that all vector fields related to the decomposition $X^{\rm RC}_3=X^{\rm 3R} _4-\sum_{j=0}^3\alpha_j(x)X^{\rm 3R}_j$ are projective vector fields: $X^{\rm 3R}_0,\ldots,X_3^{\rm 3R}$ have zero potential, while $X^{\rm 3R}_4$ have potential $-c{\rm d}u$. Therefore, all vector fields generated by $X^{\rm 3R}_0,\ldots,X_4^{\rm 3R}$ and their successive Lie brackets are projective vector fields.
\end{example}

\section{Superposition rules for the Riccati hierarchy}

Let us show that Theorem \ref{MT1} allows us to obtain a superposition rule for the members of the Riccati hierarchy by applying the superposition rule provided by Winternitz to projective Riccati equations. Winternitz {\it et al.} \cite{AHW80,AHW81} proved that a projective Riccati equation (\ref{ProRicc}) on $\mathbb{R}^n$ admits a superposition rule in terms of $n+2$ generic particular solutions of the form
$$
\Psi:\mathbb{R}^{n(n+2)}\times \mathbb{R}^n\ni (\xi_{(1)},\ldots,\xi_{(n+2)};\chi) \mapsto  \xi:=\frac{B\chi +\rho }{\langle \sigma ,\chi \rangle+b}\in \mathbb{R}^{n}\, ,
$$
where $B$ is an $n\times n$ matrix with entries $B^\mu_k:=\xi^\mu_{(k)}\sigma_k$, where no sum on $k$ is considered,
$$
\xi_{(k)}:=(\xi^1_{(k)},\ldots,\xi^n_{(k)})^T,\,\, \sigma_k:=\det(\xi_{(1)}-\xi_{(n+1)},\ldots,\stackrel{k-{\rm term}}{\overbrace{\xi_{(n+2)}-\xi_{(n+1)}}},\xi_{(n)}-\xi_{(n+1)}), \,\, k=1,\ldots,n,
$$
$B\chi$ is the matrix multiplication of the matrix $B$ with the vector $\chi:=(\chi_1,\ldots,\chi_n)^T$ that accounts for the parameters of the superposition rule, $\sigma:=(\sigma_1,\ldots,\sigma_n)^T$ and
$$
\begin{gathered}
b:=\left(1-\sum_{k=1}^n\chi_k\right)\det(\xi_{(1)}-\xi_{(n+1)},\ldots,\xi_{(n)}-\xi_{(n+1)}),\qquad \rho:=b\,\xi_{(n+2)},\\
\end{gathered}
$$
for $k,\mu=1,\ldots,n$. As usual $\langle \sigma,\chi\rangle$ is the standard inner product of $\sigma$ with $\chi$ in $\mathbb{R}^n$.

Recalling  that every $s$-order Riccati chain equation can be mapped into a projective Riccati equation on $\mathbb{R}^s$, we can immediately prove the following theorem.

\begin{theorem} Every $s$-order Riccati chain equation, when considered as a non-autonomous first-order system, admits a superposition rule depending on $s+2$ particular solutions of the form
$$
\Psi_s:[{\rm T}^{s-1}\mathbb{R}]^{s+2}\times \mathbb{R}^s\ni ({\rm t}^{s-1}u_{(1)},\ldots,{\rm t}^{s-1}u_{(s+2)},\chi) \mapsto  \phi^{-1}_{s,c}\left(\frac{B\chi +b\phi_{s,c}({\rm t}^{s-1}u_{(s+2)}) }{\langle \sigma ,\chi \rangle+b}\right)\in {\rm T}^{s-1}\mathbb{R}\, ,
$$
where $B^\mu_k:=[\phi_{s,c}({\rm t}^{s-1}u_{(k)})]^\mu\sigma_k ({\rm no\,\, sum})$ and
\begin{equation}\label{coeff}
\begin{gathered}
{b}:=\left(1-\sum_{k=1}^s\chi_k\right)\det[\phi_{s,c}({\rm t}^{s-1}u_{(1)})-\phi_{s,c}({\rm t}^{s-1}u_{(n+1)}),\ldots,\phi_{s,c}({\rm t}^{s-1}u_{(n)})-\phi_{s,c}({\rm t}^{s-1}u_{(n+1)})],\\
\sigma_k:=\det[\phi_{s,c}({\rm t}^{s-1}u_{(1)})-\phi_{s,c}({\rm t}^{s-1}u_{(n+1)}),\ldots,\stackrel{k-{\rm term}}{\overbrace{\phi_{s,c}({\rm t}^{s-1}u_{(n+2)})-\phi_{s,c}({\rm t}^{s-1}u_{(n+1)})}},\qquad\qquad\qquad\qquad\\\qquad\qquad\qquad\qquad\qquad\qquad\qquad\qquad\qquad\ldots,\phi_{s,c}({\rm t}^{s-1}u_{(n)})-\phi_{s,c}({\rm t}^{s-1}u_{(n+1)})],
\end{gathered}
\end{equation}
for  $k,\mu=1,\ldots,s$.
\end{theorem}

As an application, let us use our previous methods to obtain a superposition rule for second- and third-order Riccati chain equations. 

\begin{example} {\bf (Second-order Riccati chain equations)} The superposition rule for second-order Riccati chain equations depends on a generic family  
$$
{\rm t}^1u_{(i)}(x):=(u_{(i)}(x),v_{(i)}(x))\in {\rm T}\mathbb{R}, \qquad i=1,\ldots,4
$$
of particular solutions. The superposition rule takes the specific form
$$
\Psi_2:[{\rm T}\mathbb{R}]^{4}\times \mathbb{R}^2\ni ({\rm t}^{1}u_{(1)},\ldots,{\rm t}^{1}u_{(4)},\chi) \mapsto  {\rm t}^{1}u:=\phi^{-1}_{2,c}\left(\frac{B\chi +b\phi_{2,c}({\rm t}^{1}u_{(4)}) }{\langle \sigma ,\chi \rangle+b}\right)\in {\rm T}\mathbb{R}\, ,
$$
where $\phi_{2,c}$ is given by (\ref{phi2}) and the above-mentioned coefficients (\ref{coeff}) read 
{\small
$$
\begin{gathered}
b=\left(1-\chi_1-\chi_2\right)[(u_{(1)}-u_{(3)})(v_{(2)}-v_{(3)}+c(u^2_{(2)}-u^2_{(3)}))-(c(u^2_{(1)}-u^2_{(3)})+v_{(1)}-v_{(3)})(u_{(2)}-u_{(3)})],\\
\sigma_1=[(u_{(4)}-u_{(3)})(v_{(2)}-v_{(3)}+c(u^2_{(2)}-u^2_{(3)}))-(u_{(2)}-u_{(3)})(v_{(4)}-v_{(3)}+c(u^2_{(4)}-u^2_{(3)}))], \\
\sigma_2=[(u_{(1)}-u_{(3)})(v_{(4)}-v_{(3)}+c(u^2_{(4)}-u^2_{(3)}))-(u_{(4)}-u_{(3)})(v_{(1)}-v_{(3)}+c(u^2_{(1)}-u^2_{(3)}))], \\\qquad 
B=\left(\begin{array}{cc}u_{(1)}\sigma_1&(v_{(1)}+cu^2_{(1)})\sigma_1\\
u_{(2)}\sigma_2&(v_{(2)}+cu^2_{(2)})\sigma_2\\
\end{array}\right).\\
\end{gathered}
$$}
\end{example}

\begin{example}{\bf (Third-order Riccati chain equations)}
 The superposition rule depends on a generic family  
$$
{\rm t}^2u_{(i)}(x):=(u_{(i)}(x),v_{(i)}(x),a_{(i)}(x))\in {\rm T}^2\mathbb{R}, \qquad i=1,\ldots,5
$$ of particular solutions. The superposition rule takes the specific form
$$
\Psi_3:[{\rm T}^2\mathbb{R}]^{5}\times \mathbb{R}^3\ni ({\rm t}^{2}u_{(1)},\ldots,{\rm t}^{2}u_{(5)},\chi) \mapsto  {\rm t}^{2}u:=\phi^{-1}_{3,c}\left(\frac{B\chi +b\phi_{3,c}({\rm t}^{2}u_{(5)}) }{\langle \sigma ,\chi \rangle+b}\right)\in {\rm T}^2\mathbb{R}\, ,
$$
where the diffeomorphism $\phi_{3,c}$ is given in (\ref{phi3}) and the above-mentioned coefficients read

$$
\begin{gathered}
b=\left(1-\sum_{k=1}^3\chi_k\right)\left|\begin{array}{ccc}u_{(1)}-u_{(4)}&u_{(2)}-u_{(4)}&u_{(3)}-u_{(4)}\\
\Upsilon_{14}&\Upsilon_{24}&\Upsilon_{34}\\
\Delta_{14}&\Delta_{24}&\Delta_{34}\\
\end{array}\right|,\\
\end{gathered}
$$
where $\Upsilon_{ij}=v_{(i)}-v_{(j)}+c(u_{(i)}^2-u_{(j)}^2)$, $\Delta_{ij}:=a_{(i)}-a_{(j)}+3c(u_{(i)}v_{(i)}-u_{(j)}v_{(j)})+c^2(u_{(i)}^3-u_{(j)}^3)$ and
$$
\begin{gathered}
\sigma_1=\left|\begin{array}{ccc}u_{(5)}-u_{(4)}&u_{(2)}-u_{(4)}&u_{(3)}-u_{(4)}\\
\Upsilon_{54}&\Upsilon_{24}&\Upsilon_{34}\\
\Delta_{54}&\Delta_{24}&\Delta_{34}\\
\end{array}\right|,\quad
\sigma_2=\left|\begin{array}{ccc}u_{(1)}-u_{(4)}&u_{(5)}-u_{(4)}&u_{(3)}-u_{(4)}\\
\Upsilon_{14}&\Upsilon_{54}&\Upsilon_{34}\\
\Delta_{14}&\Delta_{54}&\Delta_{34}\\
\end{array}\right|,\\
\sigma_3=\left|\begin{array}{ccc}u_{(1)}-u_{(4)}&u_{(2)}-u_{(4)}&u_{(5)}-u_{(4)}\\
\Upsilon_{14}&\Upsilon_{24}&\Upsilon_{54}\\
\Delta_{14}&\Delta_{24}&\Delta_{54}\\
\end{array}\right|,\\
B=\left(\begin{array}{ccc}u_{(1)}\sigma_1&(v_{(1)}+cu^2_{(1)})\sigma_1&(a_{(1)}+3cu_{(1)}v_{(1)}+c^2u^3_{(1)})\sigma_1\\
u_{(2)}\sigma_2&(v_{(2)}+cu^2_{(2)})\sigma_2&(a_{(2)}+3cu_{(2)}v_{(2)}+c^2u^3_{(2)})\sigma_2\\
u_{(3)}\sigma_3&(v_{(3)}+cu^2_{(3)})\sigma_3&(a_{(3)}+3cu_{(3)}v_{(3)}+c^2u^3_{(3)})\sigma_3\\
\end{array}\right).\\
\end{gathered}
$$
\end{example}
\section{On Lie subalgebras of projective, Euclidean and hyperbolic vector fields}
Riccati chain equations were related to Vessiot--Guldberg Lie algebras of projective vector fields relative to flat Riemannian metrics in Section \ref{MT1}. These Lie algebras admit many relevant Lie subalgebras which can be additionally understood as symmetries of certain geometric structures, e.g. affine vector fields. In particular, we are interested in studying when second-order Riccati chain equations are related to Vessiot--Guldberg Lie algebras of conformal vector fields associated with Euclidean and hyperbolic metrics on the plane, which are called Euclidean and hyperbolic vector fields. To perform this study, it is necessary to analyze all Lie algebras of conformal and projective vector fields on the plane. This analysis is based on the calculus of the one-dimensional invariant distributions for all classes of finite-dimensional Lie algebras of vector fields on the plane. Although this result is new, the calculus is rather straightforward and our results have been summarized in Table \ref{table3} without further details.

\begin{lemma}\label{con} There exist no linearly independent (over $\mathbb{R}$) Euclidean vector fields $X_1, X_2$ on $\mathbb{R}^n$, for $n>1$, such that $X_1\wedge X_2=0$.
\end{lemma}
\begin{proof}  Let us prove our claim by contradiction. Assume that there exist linearly independent  (over $\mathbb{R}$) vector fields $X_1,X_2$ on $\mathbb{R}^n$ satisfying  
$$
X_1\wedge X_2=0,\qquad \mathcal{L}_{X_1}g=f_1g,\qquad \mathcal{L}_{X_2}g=f_2g,
$$
for some $f_1,f_2\in C^\infty(\mathbb{R}^n)$ and a Euclidean metric $g$. Since $X_1\wedge X_2=0$ and $X_1,X_2$ are linearly independent over $\mathbb{R}$, there exists, at least locally around any point in $\mathbb{R}^n$, a non-constant function $f$  such that $X_2=fX_1$. Therefore, we have
$$
f_2g(X,Y)=[\mathcal{L}_{X_2}g](X,Y)=[\mathcal{L}_{fX_1}g](X,Y)=f[\mathcal{L}_{X_1}g](X,Y)+g(X_1,(Xf)Y+(Yf)X),
$$
for all $X,Y\in \mathfrak{X}(\mathbb{R}^n).$ Therefore,
\begin{equation}\label{eqa}
(f_2-ff_1)g(X,Y)=g(X_1,(Xf)Y+(Yf)X),\qquad \forall X,Y\in \mathfrak{X}(\mathbb{R}^n).
\end{equation}
Since $n>1$ we can choose $X=Y\neq 0$ and $Xf=0$. Substituting these values into the above equations, we obtain $(f_2-ff_1)g(X,X)=0$.  As $g$ is Euclidean and $X\neq 0$, it follows that $g(X,X)\neq 0$ and $f_2=ff_1$. Substituting the latter in (\ref{eqa}), we obtain
\begin{equation}\label{eq11p}
0=g(X_1,(Xf)Y+(Yf)X),\qquad \forall X,Y\in \mathfrak{X}(\mathbb{R}^n).
\end{equation}
If $Xf=0$  but $Yf\neq 0$, which may occur only for $n>1$, then the above expression becomes
$$
0=g(X_1,(Yf)X)\quad \Longrightarrow \quad g(X_1,X)=0.
$$
As the above happens for every $X$ such that $Xf=0$, it is seen that $\exists h\in C^\infty(\mathbb{R}^n)\backslash\{0\}$ such that $ g(X_1,\cdot)=h{\rm d}f.$
Setting $X=Y$ and $Xf\neq 0$ in (\ref{eq11p}), we obtain 
$$
0=2(Xf)g(X_1,X)=2(Xf)^2h.
$$
This completes the proof since $h$ and $Xf$ do not vanish, the above is a contradiction and $X_1\wedge X_2\neq 0$ if $X_1,X_2$ are conformal vector fields relative to a Euclidean metric.
\end{proof}

The previous theorem does not hold on $\mathbb{R}$, e.g. $X_1:=\partial/\partial u$ and $X_2:=u\partial/\partial u$ are Euclidean linearly independent  vector fields relative to ${\rm d}u\otimes {\rm d}u$ and $X_1\wedge X_2=0$.

\begin{lemma}\label{one} Let $V$ be a Lie algebra of conformal vector fields on $N$ relative to a  metric $g$. Then,
\begin{enumerate}
 \item If the elements of $V$ leave invariant a distribution $\mathcal{D}$ on $N$, then they also leave invariant its orthogonal distribution 
$$
\mathcal{D}^\perp_\xi:=\{X_\xi\in T_\xi N:g_\xi(X_\xi,\bar X_\xi)=0,\forall \bar X_\xi\in \mathcal{D}_\xi\},\qquad \forall \xi\in N.
$$
\item If $V$ is a Lie algebra of Euclidean vector fields on $N=\mathbb{R}^2$, then it is either primitive or multi-imprimitive.
\item If $V$ is a Lie algebra of hyperbolic  vector fields on $N=\mathbb{R}^2$, then it admits two one-dimensional invariant distributions generated by two commuting vector fields $Y_1,Y_2$ and every $Z\in V$ can be brought into the form $Z=f^1_ZY_1+f^2_ZY_2$ for functions $f_Z^1,f_Z^2\in C^\infty(\mathbb{R}^2)$ satisfying $Y_1f^2_Z=Y_2f^1_Z=0$. 
\end{enumerate}

\end{lemma}
\begin{proof} Point 1): Let $X$ and $X^\perp$ be arbitrary vector fields taking values in $\mathcal{D}$ and $\mathcal{D}^\perp$, respectively. 
Hence,
\begin{equation}\label{invarDis}
0=\mathcal{L}_{Y}[g(X,X^\perp)]=f_Yg(X,X^\perp)+g(\mathcal{L}_YX,X^\perp)+g(X,\mathcal{L}_YX^\perp)=g(X,\mathcal{L}_YX^\perp),\qquad \forall Y\in V,
\end{equation}
where $f_Y$ is the potential function of $Y$. 
Then, $\mathcal{L}_{Y}X^\perp$ takes values in $\mathcal{D}^\perp$ for every $Y\in V$ and $\mathcal{D}^\perp$ is invariant under $V$. 

Point 2). If $V$ consists of Euclidean vector fields on $\mathbb{R}^2$, then any one-dimensional invariant distribution of rank one relative to $V$ satisfies $\mathcal{D}_\xi\cap \mathcal{D}_\xi^\perp=\{0\}$ for every $\xi\in N$ due to the absence of vector fields of module zero relative to $g$. The distribution $\mathcal{D}^\perp$ has rank one due to the fact that $g$ is non-degenerate. Hence, the Lie algebra has at least two different invariant distribution of rank one. Hence, $V$ is imprimitive or multiprimitive. 

Point 3). Let $\mathcal{D}_X$ be the distribution generated by a non-vanishing vector field $X$ of module zero relative to a hyperbolic metric $g$ on the plane. Since the rank of $\mathcal{D}_X^\perp$ is equal to the codimension of $\mathcal{D}_X$, namely $\dim \mathcal{D}_X^\perp=1$, and $g(X,X)=0$ by assumption, we get $\mathcal{D}^\perp_X=\mathcal{D}_X$. 
Assume that $V$ is a Lie algebra of hyperbolic vector fields and $Y\in V$. Setting $X=X^\perp$ in (\ref{invarDis}), we obtain that $\mathcal{L}_YX$ is perpendicular to $X$, hence $\mathcal{L}_YX$ takes values in $\mathcal{D}_X^\perp=\mathcal{D}_X$ and $\mathcal{D}_X$ becomes invariant under the action of $V$. 
At a fixed point, there always exist linearly independent tangent vectors with module zero relative to $g$. It is simple to prove that such tangent vectors can be extended to two well-defined vector fields $X_1,X_2$ of module zero on a neighborhood of the point spanning different distributions invariant under the action of $V$. Hence, $V$ is multi-primitive. 

Let us prove that previous invariant distributions admit two-commuting generators. Since $X_1\wedge X_2\neq 0$ and $X_1,X_2\in \mathfrak{X}(\mathbb{R}^2)$, then $[X_1,X_2]=f_1X_1+f_2X_2$ for certain functions $f_1,f_2\in C^\infty(\mathbb{R}^2)$. For arbitrary functions $h_1,h_2$, we have
$[h_1X_1,h_2X_2]=h_2(h_1f_1-X_2h_1)X_1+h_1(h_2f_2+X_1h_2)X_2$. It is trivial to show that there exist local non-vanishing solutions of  $X_2h_1=h_1f_1$ and $X_1h_2=-h_2f_2$ on an open interval of a point, e.g. as $X_1,X_2$ are non-vanishing vector fields at each point we can consider local coordinates rectifying $X_1,X_2$. Hence $Y_1:=h_1X_1,Y_2:=h_2X_2$ commute and generate the distributions $\mathcal{D}_{X_1}$ and $\mathcal{D}_{X_2}$, respectively.

Consider the distributions generated by $Y_1,Y_2$. Every vector field $X\in V$ must leave these distributions invariant. Requiring that $\mathcal{L}_{X}Y_i$ belong to the distribution $\mathcal{D}_i$ spanned by $Y_i$  and recalling that $Y_1\wedge Y_2\neq0 $, we obtain that, $X=f_X^1Y_1+f_X^2Y_2$ with $Y_2f^1_X=Y_1f^2_X=0$. 

If $V$ consists of Euclidean  conformal vector fields and it is imprimitive, then it leaves invariant the distribution perpendicular to the invariant one. So, $V$ is either multi-imprimitive or primitive.
\end{proof}

\begin{proposition}\label{ConLie} The classes ${\rm I}_1$,${\rm P}_1, {\rm P}_2
,{\rm P}_3,{\rm P}_4,{\rm P}_7,{\rm I}^{\alpha=1}_8,{\rm I}^{r=1}_{14}$ are the only classes of Lie algebras of Euclidean vector fields on $\mathbb{R}^2$. They are, up to diffeomorphism, the Lie subalgebras of {\rm P}$_7$.
\end{proposition}
\begin{proof}Lemma \ref{one} ensures that every Lie algebra of Euclidean vector fields on $\mathbb{R}^2$ must be primitive or multi-imprimitive. Moreover, Lemma \ref{con} states that there are no two linearly independent Euclidean vector fields proportional at each point. In view of Table \ref{table3}, these conditions restrict the possible classes of Lie algebras of Euclidean vector fields on the plane to
$$
{\rm I}_1,{\rm P}_1,{\rm P}_2,{\rm P}_3,{\rm P}_4,{\rm P}_7,{\rm I}_4,{\rm I}_8,{\rm I}^{r=1}_{14A},{\rm I}^{r=1}_{14B}.
$$
The classes {\rm I}$_1$, P$_1$, P$_2$, P$_3$, P$_4$, P$_7$, I$_8^{\alpha=1}$ are obviously contained in P$_7$ as is easily seen from Table \ref{table3} (see also \cite{BBHLS15}). Every Lie algebra of the class I$^{r=1}_{14}$, namely the classes I$_{14A}^{r=1}$ or I$_{14B}^{r=1}$,  can be mapped through a change of variables into a Lie subalgebra of P$_7$. Since P$_7$ is a Lie algebra of Euclidean  vector fields, all previously mentioned Lie algebras are also. Therefore, it remains only to analyze the other options: I$_4$ and I$_8$ for $\alpha\neq 1$. 

Let us prove by contradiction that I$_4$ is not a Lie algebra of Euclidean vector fields. Assume the opposite. In view of Table \ref{table3}, the Lie algebra I$_4$ leaves invariant only two one-dimensional invariant distributions $\mathcal{D}_1=\langle \partial_x\rangle$ and $\mathcal{D}_2=\langle \partial_y\rangle$. In view of  Lemma \ref{one} and the previous remark, the vector fields taking values in $\mathcal{D}_1$ and $\mathcal{D}_2$ must be orthogonal between  themselves and $g$ must take the form
$
g=g_{11}{\rm d}x\otimes {\rm d}x+g_{22}{\rm d}y\otimes {\rm d}y 
$ 
for nowhere vanishing functions $g_{xx},g_{yy}$. Since I$_4$ is by assumption a Lie algebra of Euclidean vector fields relative to $g$, it follows that it must be a Lie algebra of conformal vector fields relative to a metric
$
g_C:={\rm d}x\otimes {\rm d}x+(g_{22}/g_{11}){\rm d}y\otimes {\rm d}y. 
$
Imposing $\mathcal{L}_Xg_C=f_Xg$ for every $X\in {\rm I}_4$, it is easily seen that $g_{11}/g_{22}=0$ and $g_C$ is not a metric. Hence, I$_4$ is not a Lie algebra of Euclidean vector fields. Similarly, it can be proved that I$^{\alpha\neq 1}_8$ is not a Lie algebra of Euclidean vector fields.

All Lie subalgebras of vector fields of P$_7$ consist of Euclidean vector fields. Therefore, they all must be among the described in our proposition. As mentioned above, all previous vector fields can be mapped through a change of variables into Lie subalgebras of P$_7$. Hence, P$_7$ contains all finite-dimensional Lie algebras of conformal vector fields on the plane with respect to a Euclidean metric.
\end{proof}
\begin{proposition}\label{ProjEuc} Every Lie algebra of Euclidean and projective vector fields on $\mathbb{R}^2$ is diffeomorphic to a Lie subalgebra of {\rm P}$_{2}$, {\rm P}$_3$ or {\rm P}$_4$.
\end{proposition}
\begin{proof} Lemma  \ref{ConLie} characterizes  all classes of finite-dimensional Lie algebras of Euclidean vector fields on $\mathbb{R}^2$. To prove our result, it is necessary to determine which of them are also diffeomorphic to Lie subalgebras of P$_8$. Finally, we show that all Lie algebras that are projective and Euclidean are diffeomorphic to Lie subalgebras of {\rm P}$_{2}$, {\rm P}$_3$ or {\rm P}$_4$.

Let us check all the Lie algebras of Euclidean vector fields given in Proposition \ref{ConLie}.
Table \ref{table3} shows that I$_1$, P$_1$, P$_4$, I$_{8}^{\alpha=1}$, I$^{r=1}_{14B}$, I$^{r=1}_{14A}$ are Lie subalgebras of P$_8$ and therefore they consist of projective vector fields. Since P$_7\simeq \mathfrak{so}(3,1)$ and P$_8\simeq \mathfrak{sl}(3)$, it follows that P$_7$ is not isomorphic to any Lie subalgebra of $\mathfrak{sl}(3)$. It remains to verify whether P$_2$ and P$_3$ can be mapped into a Lie subalgebra of P$_8$. 

Let us prove that P$_3$ is diffeomorphic to a Lie subalgebra of P$_8$. Observe that P$_8\simeq\mathfrak{sl}(3)$ admits a Lie subalgebra isomorphic to $\mathfrak{so}(3)$. In view of Table \ref{table3}, there exists only one finite-dimensional Lie algebra of vector fields on the plane isomorphic to $\mathfrak{so}(3)$, namely P$_3$. Then P$_3$ must be diffeomorphic to a Lie subalgebra of P$_8$ and it becomes a Lie algebra of Euclidean projective vector fields.

Now, we show that P$_8$ admits a Lie subalgebra diffeomorphic to P$_2$. Consider the Lie subalgebra of P$_8$ spanned by the vector fields (cf. Table \ref{table3})
$$
Y_1=-xy\partial_x+(2x-y^2)\partial_y,\qquad Y_2=-2x\partial_x-y\partial_y,\qquad
Y_3=-y\partial_x-2\partial_y.
$$
Since $[Y_1,Y_2]=Y_1,[Y_1,Y_3]=2Y_2,[Y_2,Y_3]=Y_3$, we have $\langle Y_1,Y_2,Y_3\rangle \simeq \mathfrak{sl}(2)$. Using the formalism developed in \cite{BHLS15}, the above Lie algebra can be locally mapped into P$_2$ around a point of $\mathbb{R}^2$ if and only if the sign of the determinant of the coefficients of the tensor field $\mathcal{R}=Y_1\otimes Y_3+Y_3\otimes Y_1-2Y_2\otimes Y_2$ is positive around such a point. 
 In our case, we have that
$$
\mathcal{R}=2x(y^2-4x)\partial_x\otimes\partial_x+2(y^2-4x)\partial_y\otimes\partial_y+y(y^2-4x)(\partial_x\otimes\partial_y+\partial_y\otimes\partial_x)\Rightarrow \det \mathcal{R}=(4x-y^2)^3.
$$
Hence, the above Lie algebra is locally diffeomorphic to P$_2$ if and only if $4x>y^2$. Hence, P$_2$ can be considered as a Lie algebra of Euclidean and projective vector fields. 

Let us finally show that previous Lie algebras are diffeomorphic to Lie subalgebras of P$_2$, P$_3$ and P$_4$. It was proved in \cite{BBHLS15} that I$_{14A}$, I$_{14B}$, I$^{\alpha=1}_8$ and  P$_1$ are contained in P$_4$.
Since P$_2$ and P$_3$ are simple and P$_4$ is solvable, it follows that P$_2$ and P$_3$ are not isomorphic to any Lie subalgebra of P$_4$. Hence, any Lie algebra of Euclidean projective vector fields on the plane is diffeomorphic to a Lie subalgebra of P$_2$, P$_3$ or P$_4$. 
\end{proof}

\begin{proposition}\label{Hyp} A Vessiot--Guldberg Lie algebra on $\mathbb{R}^2$ consists of hyperbolic vector fields if and only if it is diffeomorphic to a Lie subalgebra of {\rm I}$_{11}$, namely ${\rm I}_1-{\rm I}_4,{\rm I}_6,{\rm I}_8,{\rm I}_9-{\rm I}_{11}, {\rm I}^{r=1}_{14B}, {\rm I}^{r=1}_{15B}.$
\end{proposition}
\begin{proof} Lemma \ref{one} states that every Lie algebra $V$ of hyperbolic vector fields on $\mathbb{R}^2$ admits, at least, two different invariant distributions spanned by two commuting vector fields $Y_1,Y_2$ and
	every  $Z\in V$ can be brought into the form $Z=f_Z^1Y_1+f_Z^2Y_2$ for some  $f_Z^1,f_Z^2\in C^\infty(\mathbb{R}^2)$ with $Y_2f^1_Z=Y_1f^2_Z=0$.
	In view of Table \ref{table3}, the only options are those given in the corollary. Additionally, it can be proved immediately that these Lie algebras are Lie subalgebras of I$_{11}$ and they are therefore diffeomorphic to a Lie algebra of hyperbolic vector fields.
\end{proof}

\begin{proposition}\label{HypPro} A Vessiot--Guldberg Lie algebra of vector fields on $\mathbb{R}^2$ consists of projective hyperbolic vector fields if and only if it is diffeomorphic to a Lie subalgebra of {\rm I}$_4$ or {\rm I}$_9$.
\end{proposition}
\begin{proof}

Proposition \ref{Hyp} establishes that the Lie algebras of hyperbolic vector fields are the Lie subalgebras of I$_{11}$. Let us determine which of them are diffeomorphic to Lie subalgebras of P$_8$. Such algebras can also be considered as Lie algebras of projective vector fields. Many of the Lie subalgebras of I$_{11}$ can be proved not to be diffeomorphic to a Lie subalgebra of P$_8$ using the following argument. A long but straightforward computation shows that the only three-dimensional Lie subalgebra of P$_8$  spanning a one-dimensional distribution is spanned by the vector fields
$$
x\partial_x+y\partial_y,\quad x(x\partial_x+y\partial_y),\quad y(x\partial_x+y\partial_y).
$$
Since the last two vector fields commute among themselves, these vector fields do not span a Lie algebra isomorphic to $\mathfrak{sl}(2)$. So, there is no Lie subalgebra isomorphic to $\mathfrak{sl}(2)$ in P$_8$ spanning a distribution of rank one. As a consequence, I$_3$ is not diffeomorphic to any Lie subalgebra of P$_8$.   
all other Lie subalgebras of I$_{11}$ containing a Lie subalgebra diffeomorphic to I$_3$, namely I$_6$, I$_{10}$ and I$_{11}$ as seen in Table \ref{table3}, are not  diffeomorphic to any Lie subalgebra of P$_8$ either.

The remaining Lie subalgebras of I$_{11}$ can be shown to be diffeomorphic to Lie subalgebras of P$_8$. In view of Table \ref{table3}, the Lie algebras I$_1$, I$_2$, I$_8$, I$_9$, I$^{r=1}_{14B}$, I$^{r=1}_{15B}$ are trivially contained in P$_8$. It was shown in the proof of Proposition \ref{ProjEuc} that I$_4$ is locally diffeomorphic to one of the Lie subalgebras of P$_8$. The Lie algebras I$^{r=1}_{14A}$ and I$^{r=1}_{15B}$ can be shown to be diffeomorphic to Lie subalgebras of P$_8$ using the change of variables $\bar y=e^{-cx}y,\bar x=x$.

All above-mentioned Lie algebras are diffeomorphic to Lie subalgebras of I$_4$ or I$_9$. As I$_4$ cannot be contained in I$_9$ because I$_9$ does not contain any Lie subalgebra isomorphic to $\mathfrak{sl}(2)$, it follows that every Lie algebra of projective hyperbolic vector fields on the plane is locally diffeomorphic to a Lie subalgebra of I$_4$ or I$_9$.

\end{proof}

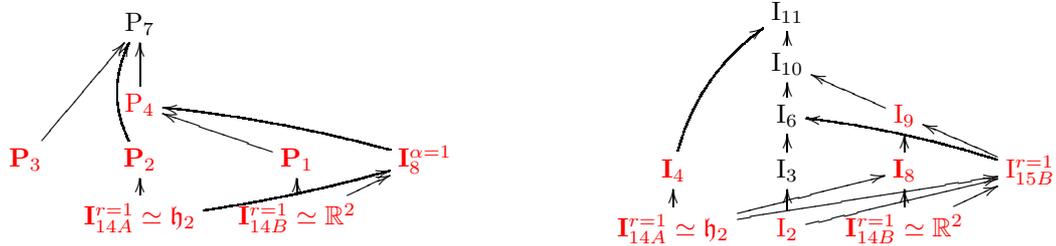
\begin{table} [h!]{\footnotesize
 \noindent
\caption{{\small All Lie algebras of Euclidean and hyperbolic vector fields on $\mathbb{R}^2$ according to Propositions from \ref{ConLie} to \ref{HypPro}  and their inclusion relations. Lie algebras, which can also be considered as Lie algebras of projective vector fields, are highlighted in red. Arrows indicate all inclusion relations among Lie algebras. Such inclusion relations can easily be obtained after a long but simple computation. In bold are highlighted those Lie algebras of Hamiltonian vector fields relative to a symplectic structure (see \cite{BBHLS15} for details)}}
\label{table5}
\begin{minipage}{8cm}
\qquad\qquad\xymatrix@R=1.8mm@C=3.4mm{
                                         &{\rm P}_7    &      &  \\
                                         &                    &      &\\
                                         &\color{red}{\rm P}_4\ar[uu]    &     & \\
\color{red}{\rm \bf P}_3\ar[ruuu]    &\color{red}{\rm \bf P}_2\ar@/^{3mm}/[uuu]&\color{red}{\rm \bf P}_1\ar[lu]&\color{red}{\rm \bf I}^{\alpha=1}_8\ar@/_{1mm}/[ull]\\
                                        &  {\color{red}{\bf I}^{r=1}_{14A}\simeq \mathfrak{h}_2}\ar[u] \ar@/_{1mm}/[urr]            &{\color{red}{\rm \bf I}^{r=1}_{14B}\simeq \mathbb{R}^2\ar[ru]\ar[u]}}\end{minipage}\begin{minipage}{8cm}\qquad\qquad\xymatrix@R=1.8mm@C=3.4mm{
                                     &{\rm I}_{11}    &      &  \\
                                        & {\rm I}_{10}\ar[u]                   &      &\\
                                         &{\rm I}_6\ar[u]    &\color{red} {\rm I}_9\ar[lu]    & \\
 \color{red} {\rm \bf I}_4\ar@/^{3mm}/[uuur]&{\rm I}_3\ar[u]&\color{red}{\rm \bf I}_8\ar[u]&\color{red}{\rm  I}^{r=1}_{15B}\ar[ul]\ar@/_{1mm}/[ull]\\
                                        {\color{red}{\rm \bf I}^{r=1}_{14A}\simeq \mathfrak{h}_2\ar[rrru]\ar[rru]\ar[u]}&\color{red}{\rm I}_2\ar[u]\ar[urr]              &{\color{red}{\rm \bf I}^{r=1}_{14B}\simeq \mathbb{R}^2\ar[ru]\ar[u]}                        \\ 
}\end{minipage}}
\end{table}

\begin{table}[h!] {\footnotesize
 \noindent
\caption{{\small {Non-exhaustive tree of inclusion relations} between classes of the GKO classification. The diagram details all Lie subalgebras of I$_{11}$ and P$_7$.
We write $A\rightarrow 
B$ when a subclass of $A$ is diffeomorphic to a Lie subalgebra of $B$. Every Lie algebra includes I$_1$. In bold and italics are classes with Hamiltonian Lie algebras and rank one associated distributions, respectively. Colors help to distinguish the arrows.}}
\label{table2}
\medskip
\xymatrix@R=1.8mm@C=3.4mm{
 \dim\,>6             &             &              &{\rm P}_8&                         &&&&&{\rm I}_{20}\ar@{-->}[llllll]_{r=1}&&\\
 \dim\,6\rightarrow   &       &          {\rm P}_7 &{\rm P}_6\ar[u]         &          {\rm I}_{11}&&&         & {\rm I}_{19}\ar[ru]  &{\rm I}_{18}\ar@/_{4mm}/@{-->}[llllll]^{r=1}\ar[u]&      &\\
\dim\,5\rightarrow&          &           &{\rm \bf P}_5\ar[u]          &{\rm I}_{10}\ar[u]&&   &   &{\rm \bf I}_{16}\ar@/_{3mm}/@{-->}[lllll]_{\,\,\,\quad\alpha=-1,r=1}\ar@{--}[u]\ar[ru] &&{\rm I}_{15}\ar@{--}[ul]& & \\
\dim\,4\rightarrow& &{\rm P}_4\ar[ruu]\ar[uu]&     &{\rm I}_6\ar[u] &{\rm I}_9\ar@/^{1mm}/[lluu]\ar@/_{1mm}/[rrrruu]\ar[lu]&{\rm I}_7\ar@/^{4mm}/[uuurrr]\ar@/_{2mm}/[uulll]      &&{\rm I}_{17}\ar[u]&{\color{black} {\it I}_{13}}\ar@{--}[ur]&{\rm I}_{14}^{r> 1}\ar@{-->}@/_{2mm}/[lllllllu]\ar[u]&&\\
\dim\,3\rightarrow  &{\rm \bf P}_3\ar@/^{10mm}/[uuuurr]\ar[ruuu]    &{\rm \bf P}_2\ar@/_{1mm}/[uuuur]\ar@/^{3mm}/[uuu]&{\rm \bf P}_1\ar@{--}[uu]\ar[lu]&{\color{black}{\it I}_3}\ar[u]
&{\rm \bf I}_8\ar@{--}@/^{0mm}/[ulll]^{\alpha=1}\ar@{--}@/^{1mm}/[uull]_{\alpha=-1\,\,\,\,\,\,\,\,\,\,\,\,\,\,\,\,\,\,\,\,\,\,\qquad}\ar@/_{1mm}/[uurrr]\ar[u]\ar@{--}@/_{1mm}/[uuurrr]&{\rm \bf I}_5\ar@/_{0mm}/[uulll]\ar[u]\ar[uuurr]      &{\rm \bf I}_4\ar@/_{4mm}/[llluuu]\ar@/_{5mm}/[lllluuuu]&&{\color{black}{\it \bf I}_{12}^{r>1}}\ar@{--}[ul]\ar@{--}[ur]\ar[u]&&&\\
&&&&&&&&&& &\\
\dim\,2\rightarrow          &          &&  {\color{brown}{\bf I}^{r=1}_{14A}\simeq \mathfrak{h}_2}\ar@[brown]@/_{13mm}/[rrrrrrruuu]\ar@[brown]@/^{2mm}/[rruu]\ar@[brown]@/^{2mm}/[rrruu]\ar@[brown]@/_{2mm}/[rrrrruuu]\ar@[brown]@/^{2mm}/[rrrruu]\ar@[brown][luu]\ar@[brown]@/_{2mm}/[uuul]               &{\color{orange}{\it I}_2\simeq \mathfrak{h}_2}\ar@[orange]@/_{12mm}/[rrrrrruuuu]\ar@[orange]@/_{4mm}/[rrrruuuu]\ar@[orange]@/_{4mm}/[rrrrruuu]\ar@[orange][uu]\ar@[orange][uuur]\ar@[orange][uuurr]&&{\color{blue}{\rm \bf I}^{r=1}_{14B}\simeq \mathbb{R}^2\ar@[blue]@/_{1mm}/[uuurrrr]\ar@[blue]@/_{4mm}/[uuull]\ar@[blue]@/_{4mm}/[uuu]\ar@/^{2mm}/@[blue][llluuuu]\ar@[blue]@/^{2mm}/[llluu]\ar@[blue][rruuuuu]\ar@[blue]@/_{3mm}/[rruuu]\ar@[blue][luu]}& {\color{red}{\bf {\it I}}_{\bf 12}^{\bf r=1}\simeq\mathbb{R}^2}\ar@/^{1mm}/@[red][lllluuuu]\ar@/^{2mm}/@[red][ruuuuu]\ar@[red][ruuu]\ar@/_{2mm}/@[red][rruu]\ar@[red]@/_{7mm}/[rrruuu]&          &                     &&&&\\ 
\save "2,1"."7,8"*\frm<8pt>{.}\restore
\save "1,1"."7,12"**\frm<8pt>{--}\restore
}}
\end{table}

\begin{table}[h!] {\tiny
 \noindent
\caption{{\small We here describe the so-called GKO classification of  the $8+20$  finite-dimensional real Lie algebras of vector fields on the plane and their most relevant
features. The first (one or two) vector fields which are written between brackets form a modular generating system.
The functions $\xi_1(x),\ldots,\xi_r(x)$ and $1$ are linearly independent and the functions $\eta_1(x),\ldots,\eta_r(x)$ form a basis of fundamental solutions for an $r$-order homogeneous differential equation with constant coefficients \cite[pp.~470--471]{GKP92}. Finally, $\mathfrak{g}=\mathfrak{g}_1\ltimes \mathfrak{g}_2$ means that $\mathfrak{g}$ is the direct sum (as linear subspaces) of $\mathfrak{g}_1$ and $\mathfrak{g}_2$, where $\mathfrak{g}_2$ is an ideal of $\mathfrak{g}$.}}
\label{table3}
\medskip
\noindent\hfill
 \begin{tabular}{ p{.2cm} p{2cm}    p{8.5cm} p{1cm}l}
\hline
&  &\\[-1.9ex]
\#&Primitive & Basis of vector fields $X_i$ &Domain&Inv. Dis\\[+1.0ex]
\hline
 &  &\\[-1.9ex]
P$_1$&$A_\alpha\simeq \mathbb{R}\ltimes \mathbb{R}^2$ & $  { {\partial_x} ,    {\partial_y} ,   \alpha(x\partial_x + y\partial_y)  +  y\partial_x - x\partial_y},\quad \ \alpha\geq 0$&$\mathbb{R}^2$&$-$ \\[+1.0ex]
P$_2$&$\mathfrak{sl}(2)$ & $ {\partial_x},   {x\partial_x  +  y\partial_y} ,   (x^2  -  y^2)\partial_x  +  2xy\partial_y$&$\mathbb{R}^2_{y\neq 0}$&$-$\\[+1.0ex]
P$_3$&$\mathfrak{so}(3)$ &${     { y\partial_x  -  x\partial _y},     { (1  +  x^2  -  y^2)\partial_x  +  2xy\partial_y} ,   2xy\partial_x  +  (1  +  y^2  -  x^2)\partial_y}$&$\mathbb{R}^2$&$-$\\[+1.0ex]
P$_4$&$\mathbb{R}^2\ltimes\mathbb{R}^2$ &$  {\partial_x},   {\partial_y},  x\partial_x + y\partial_y,   y\partial_x - x\partial_y$&$\mathbb{R}^2$&$-$\\[+1.0ex]
P$_5$&$\mathfrak{sl}(2 )\ltimes\mathbb{R}^2$ &${  {\partial_x},   {\partial_y},  x\partial_x - y\partial_y,  y\partial_x,  x\partial_y}$&$\mathbb{R}^2$&$-$\\[+1.0ex]
P$_6$&$\mathfrak{gl}(2 )\ltimes\mathbb{R}^2$ &${ {\partial_x},    {\partial_y},   x\partial_x,   y\partial_x,   x\partial_y,   y\partial_y}$&$\mathbb{R}^2$&$-$\\[+1.0ex]
P$_7$&$\mathfrak{so}(3,1)$ &${  {\partial_x},   {\partial_y},   x\partial_x\!+\! y\partial_y,   y\partial_x \!-\! x\partial_y,   (x^2 \!-\! y^2)\partial_x \!+\! 2xy\partial_y,  2xy\partial_x \!+\! (y^2\!-\!x^2)\partial_y}$  &$\mathbb{R}^2$&$-$\\[+1.0ex]
P$_8$&$\mathfrak{sl}(3 )$ &${  {\partial_x},    {\partial_y},   x\partial_x,   y\partial_x,   x\partial_y,   y\partial_y,   x^2\partial_x + xy\partial_y,   xy\partial_x  +  y^2\partial_y}$&$\mathbb{R}^2$&$-$\\[+1.5ex]
\hline
&  &\\[-1.5ex]
\#& Mono-imprimitive\!\! & Basis of vector fields $X_i$ &Domain&Inv. Dis\\[+1.0ex]
\hline
&  &\\[-1.5ex]
I$_5$&$\mathfrak{sl}(2 )$ &${ {\partial_x},    {2x\partial_x + y\partial_y},   x ^2\partial_x  +  xy\partial_y}$&$\mathbb{R}^2_{y\neq 0}$&$\partial_y$\\[+1.0ex]
I$_7$&$\mathfrak{gl}(2 )$ & ${  {\partial_x},   {y\partial_y} ,     x\partial_x,    x^2\partial_x +  xy \partial_y}$&$\mathbb{R}^2_{y\neq 0}$ &$\partial_y$\\[+1.0ex]
I$_{12}$&$\mathbb{R}^{r + 1}$ &$ {\partial_y} ,   \xi_1(x)\partial_y, \ldots , \xi_r(x)\partial_y,\quad   r\geq 1$&$\mathbb{R}^2$ &$\partial_y$\\[+1.0ex]
I$_{13}$&$\mathbb{R}\ltimes \mathbb{R}^{r + 1}$ &$  {\partial_y} ,   y\partial_y,    \xi_1(x)\partial_y, \ldots , \xi_r(x)\partial_y,\quad   r\geq 1$ &$\mathbb{R}^2$&$\partial_y$\\[+1.0ex]
I$_{14}$&$\mathbb{R}\ltimes \mathbb{R}^{r}$ & ${ {\partial_x},   {\eta_1(x)\partial_y} ,  {\eta_2(x)\partial_y},\ldots ,\eta_r(x)\partial_y},\quad (r> 1, \eta_1'(x)\neq\eta_1(x))$&$\mathbb{R}^2$&$\partial_y$\\[+1.0ex]
I$_{15}$&$\mathbb{R}^2\ltimes \mathbb{R}^{r}$ &  ${ {\partial_x},    {y\partial_y} ,    {\eta_1(x)\partial_y},\ldots, \eta_r(x)\partial_y},\quad  (r> 1, \eta_1'(x)\neq\eta_1(x))$&$\mathbb{R}^2$&$\partial_y$\\[+1.0ex]
I$_{16}$&$C_\alpha^r\!\simeq\! \mathfrak{h}_2\!\ltimes\!\mathbb{R}^{r + 1}$ & ${  {\partial_x},    {\partial_y} ,   x\partial_x  +  \alpha y\partial y,   x\partial_y, \ldots, x^r\partial_y},\quad   r\geq 1,\qquad \alpha\in\mathbb{R}$&$\mathbb{R}^2$&$\partial_y$\\[+1.0ex]
I$_{17}$&$\mathbb{R}\ltimes(\mathbb{R}\ltimes \mathbb{R}^{r})$ &$  {\partial_x},    {\partial_y} ,   x\partial_x  +  (ry  +  x^r)\partial_y ,   x\partial_y, \ldots,  x^{r - 1}\partial_y,\quad   r\geq 1$ &$\mathbb{R}^2$&$\partial_y$\\[+1.0ex]
I$_{18}$&$(\mathfrak{h}_2\!\oplus\! \mathbb{R})\!\ltimes\! \mathbb{R}^{r + 1}$ & $  {\partial_x},    {\partial_y},   x\partial_x,   x\partial_y,   y\partial_y,   x^2\partial_y, \ldots,x^r\partial_y,\quad r\geq 1$ &$\mathbb{R}^2$&$\partial_y$\\[+1.0ex]
I$_{19}$&$\mathfrak{sl}(2 )\ltimes \mathbb{R}^{r + 1}$ &  $  {\partial_x},    {\partial_y} ,   x\partial_y,    2x\partial _x  +  ry\partial_y,   x^2\partial_x  +  rxy\partial_y,   x^2\partial_y,\ldots, x^r\partial_y ,\quad   r\geq 1$&$\mathbb{R}^2$ &$\partial_y$\\[+1.0ex]
I$_{20}$&$\mathfrak{gl}(2 )\ltimes \mathbb{R}^{r + 1}$ &  $  {\partial_x},    {\partial_y} ,   x\partial_x,   x\partial_y,   y\partial_y,   x^2\partial_x  +  rxy\partial_y,   x^2\partial_y,\ldots, x^r\partial_y,\quad   r\geq 1$ &$\mathbb{R}^2$&$\partial_y$\\[1.5ex]
\hline
&  &\\[-1.5ex]
\#& Multi-imprimitive\!\! & Basis of vector fields $X_i$ &Domain&Inv. Dis\\[+1.0ex]
\hline
&  &\\[-1.5ex]
I$_1$&$\mathbb{R}$ &$  {\partial_x} $ & $\mathbb{R}^2$&$\partial_y,\partial_x+h(y)\partial_y$\\[+1.0ex]
I$_2$&$\mathfrak{h}_2$ & $  {\partial_x} ,  x\partial_x$& $\mathbb{R}^2$&$\partial_x,\partial_y$\\[+1.0ex]
I$_3$&$\mathfrak{sl}(2 )$  &$  {\partial_x},  x\partial_x,  x^2\partial_x$& $\mathbb{R}^2$&$\partial_x,\partial_y$\\[+1.0ex]
I$_4$&$\mathfrak{sl}(2 )$  & ${  {\partial_x  +  \partial_y},    {x\partial _x + y\partial_y},   x^2\partial_x  +  y^2\partial_y}$ &$\mathbb{R}^2_{x\neq y}$&$\partial_x,\partial_y$\\[+1.0ex]
I$_6$&$\mathfrak{gl}(2 )$ & ${ {\partial_x},    {\partial_y},   x\partial_x,   x^2\partial_x}$&$\mathbb{R}^2$&$\partial_x,\partial_y$\\[+1.0ex]
I$^{\alpha\neq 1}_8$&$B_{\alpha\neq 1}\simeq \mathbb{R}\ltimes\mathbb{R}^2$ &${  {\partial_x},    {\partial_y},   x\partial_x  +  \alpha y\partial_y},\quad  0<|\alpha|\leq 1,\alpha\neq 1$&$\mathbb{R}^2$&$\partial_x,\partial_y$\\[+1.0ex]
I$^{\alpha=1}_8$&$B_1\simeq \mathbb{R}\ltimes\mathbb{R}^2$ &${  {\partial_x},    {\partial_y},   x\partial_x  +  y\partial_y}$&$\mathbb{R}^2$&$\lambda_x\partial_x+\lambda_y\partial_y$\\[+1.0ex]
I$_9$&$\mathfrak{h}_2\oplus\mathfrak{h}_2$ &${ {\partial_x},    {\partial_y},   x\partial_x,  y\partial_y}$&$\mathbb{R}^2$&$\partial_x,\partial_y$\\[+1.0ex]
I$_{10}$&$\mathfrak{sl}(2 )\oplus \mathfrak{h}_2$ & ${ {\partial_x},    {\partial_y} ,   x\partial_x,  y\partial_y,  x^2\partial_x }$&$\mathbb{R}^2$&$\partial_x,\partial_y$\\[+1.0ex]
I$_{11}$&$\mathfrak{sl}(2 )\oplus\mathfrak{sl}(2 )$ &$  {\partial_x},    {\partial_y},   x\partial_x,   y\partial_y,   x^2\partial_x ,   y^2\partial_y $&$\mathbb{R}^2$&$\partial_x,\partial_y$\\[+1.0ex]
I$^{r=1}_{14A}$&$\mathbb{R}\ltimes \mathbb{R}$ & $ {\partial_x},    { e^{cx}\partial_y},\quad  c\in \mathbb{R}\backslash0$&$\mathbb{R}^2$&$e^{cx}\partial_y,\partial_x+cy\partial_y$\\[+1.0ex]
I$^{r=1}_{14B}$&$\mathbb{R}\ltimes \mathbb{R}$ & $ {\partial_x},    {\partial_y}$&$\mathbb{R}^2$&$\lambda_x\partial_x+\lambda_y\partial_y$\\[+1.0ex]
I$^{r=1}_{15A}$&$\mathbb{R}^2\ltimes \mathbb{R}$ &  $ {\partial_x},    {y\partial_y} ,    {e^{cx}\partial_y},\quad  c \in \mathbb{R}\backslash0$&$\mathbb{R}^2$&$e^{cx}\partial_y,\partial_x+cy\partial_y$\\[+1.0ex]
I$^{r=1}_{15B}$&$\mathbb{R}^2\ltimes \mathbb{R}$ &  $ {\partial_x},    {y\partial_y} ,    \partial_y$&$\mathbb{R}^2$&$\partial_x,\partial_y$\\[+1.0ex]
\hline
 \end{tabular}
\hfill}
\end{table}

\section{Second-order Riccati chain equations and conformal Riccati equations}

The findings of the previous section allow us to characterize when second-order Riccati chain equations can be related to hyperbolic and/or Euclidean Riccati equations. In this case, superposition rules depending on fewer particular solutions are available \cite{AW80}.
To prove these results, we make use of the following Lemmas \ref{8Lie} and \ref{RiccNoHam}.

\begin{lemma}\label{8Lie} The Lie algebra {\rm P}$_4$ is the only four-dimensional solvable Lie algebra $V$ of vector fields on $\mathbb{R}^2$ such that: a) the vector fields of the ideal $[V,V]$ span a distribution of rank two and $\dim [V,V]=2$ and b) $V$ acts irreducibly on  $[V,V]$ via the adjoint representation.
\end{lemma}

\begin{proof}  
In view of Table \ref{table3}, there exist six classes of solvable four-dimensional Lie algebras of vector fields on $\mathbb{R}^2$ spanning a distribution of rank two: P$_4$, I$_9$, I$_{14}$ for $r=3$, I$_{15}$ for $r=2$, I$_{16}$ for $r=1$ and I$_{17}$ for $r=2$. If we demand that their first derived Lie algebra  be two-dimensional and span a distribution of rank two, then the above list reduces to P$_4$,  I$_9$, I$_{16}$ for $r=1$ and $\alpha=1$. All of  these Lie algebras are of the form $\mathbb{R}^2\ltimes \mathbb{R}^2$, where the ideal is the first-derived series of the Lie algebra. Let us determine how elements of P$_4$ act on $[{\rm P}_4,{\rm P}_4]$ irreducibly under the adjoint representation. 

We see in   Table \ref{table3} that  P$_4$ admits a basis $X_1,X_2,X_3,X_4$. Let us fix a basis  $\mathcal{B}:=\{X_1,X_2\}$ for the first derived Lie algebra.
 The vector field $Z:=aX_3+bX_4+cX_1+dX_2$, with $a,b,c,d\in\mathbb{R}$,  acts on the Abelian ideal $\langle X_1,X_2\rangle=[{\rm P}_4,{\rm P}_4]$ as a morphism having  the following matrix in the chosen basis 
\begin{equation}\label{Pattern}
A_Z:=[{\rm ad}_Z]^\mathcal{B}_\mathcal{B}=\left(\begin{array}{cc}
-a&-b\\
b&-a
\end{array}\right).
\end{equation}
Therefore, P$_4$ acts on $[{\rm P}_4,{\rm P}_4]$ irreducibly (over $\mathbb{R}$). Meanwhile, we note by inspecting Table \ref{table3} that the Lie algebra I$_9$ admits a one-dimensional ideal spanned by $X_1\in [{\rm I}_9,{\rm I}_9]$ 
 and  the Lie algebra I$_{16}$ with $\alpha=1$ and $r=1$ admits a one-dimensional ideal spanned by $X_2\in [{\rm I}_{16},{\rm I}_{16}]$. Hence, I$_9$ and I$_{16}$, with $\alpha=1$ and $r=1$, do not act irreducibly on their first derived ideals. This finishes our proof.
 
 \end{proof}

Second-order Riccati chain equations are hereafter written as a first-order system related to an  $x$-dependent vector field  $X_2^{\rm RC}$. The irreducible Lie algebra $V_2$ of $X_2^{\rm RC}$ is spanned by linear combinations of the vector fields $\Gamma_0:=(X_2^{\rm RC})_{x_0}$ and $\{\Delta_{x,x_0}:=(X_2^{\rm RC})_x-(X_2^{\rm RC})_{x_0}\}_{x\in\mathbb{R}}$ and their successive Lie brackets for arbitrary $x_0\in \mathbb{R}$. Although $\Delta_{x,x_0}$ depends on the arbitrarily chosen $x,x_0$, the linear space $\Delta:=\langle \Delta_{x,x_0},x\in \mathbb{R}\rangle$ has an intrinsic meaning independent of them. We say that a second-order Riccati chain equation is {\it strictly non-autonomous} if $\Delta \neq 0$, i.e. $\Delta_{x,x_0}\neq 0$ for some $x,x_0\in\mathbb{R}$. The elements of $\Delta$ related to a strictly non-autonomous second-order Riccati chain equation span a generalized distribution of rank at most one at any point of ${\rm T}\mathbb{R}$. 

To simplify our presentation, we assume in what follows that $\Delta\neq 0$ and $X_2^{\rm RC}$ is written in the coordinate system $\{\xi_1:=x,\xi_2:=v+cu^2\}$ on $\mathbb{R}^2$ which allows us to consider $X_2^{\rm RC}$ in the form of a projective equation on $\mathbb{R}^2$. Additionally, we set $\partial_i:=\partial/\partial \xi_i$  for $i=1,2$. Observe that $\Delta_{x,x_0}=P(x,\xi){\partial_2}$ 
for a certain polynomial $P(x,\xi)$ in the variables $\xi_1,\xi_2$ with $x$-dependent coefficients. Evidently, $\Delta_{x_1,x_0}\wedge\Delta_{x_2,x_0}=0$ for arbitrary $x_0,x_1,x_2\in \mathbb{R}$. Moreover, we define $\Gamma_0:=(X^{\rm RC}_2)_{x_0}$ and $\Gamma_1:=\Delta_{x_1,x_0}$ for $x_1,x_0$ such that $\Delta_{x_1,x_0}\neq 0$. Hence, there exist  constants $c_0,c_1,c_2,\bar{c}_0,\bar{c}_1,\bar{c}_2\in \mathbb{R}$ such that  
\begin{equation}\label{Gen}
\Gamma_0:=\left[\bar{c}_0+\sum_{\alpha=1}^2\bar{c}_\alpha\xi_\alpha\right]\partial_2+\xi_2\partial_{1}-c\xi_1\sum_{\alpha=1}^2\xi_\alpha\partial_\alpha,\qquad
\Gamma_1:=\left[c_0+\sum_{\alpha=1}^2c_\alpha\xi_\alpha\right]\partial_2\neq0.
\end{equation}


\begin{lemma}\label{RiccNoHam} Strictly non-autonomous second-order Riccati chain equations, written as a first-order system, do not admit any Vessiot--Guldberg Lie algebra of locally Hamiltonian vector fields relative to a symplectic structure on ${\rm T}\mathbb{R}$. A strictly non-autonomous second-order affine Riccati chain equation is Hamiltonian if and only if any pair of vector fields $\Gamma_0,\Gamma_1$ satisfies one of the following conditions
\begin{equation}\label{AffRiccHam}
a) \,\,c_2=\bar c_2=0,\qquad b)\,\, c_1=\bar c_1=0, c_2\bar c_0-c_0\bar c_1c_
2=0.
\end{equation}
\end{lemma}
\begin{proof} Let us proceed by contradiction. Assume that the irreducible Lie algebra $V_2$ associated with our strictly non-autonomous second-order chain Riccati equation consists of Hamiltonian vector fields relative to a symplectic structure.  Then, $\Gamma_0$ and $\Gamma_1$  must be Hamiltonian vector fields relative to a symplectic structure $\Omega:=f{\rm d}\xi_1\wedge {\rm d}\xi_2$ on $\mathbb{R}^2$ (see Table \ref{table2} or cf. \cite{BBHLS15}). This amounts to the fact $\mathcal{L}_{\Gamma_0}\Omega=\mathcal{L}_{\Gamma_1}\Omega=0$ for a non-vanishing function $f$. In coordinates, these conditions read 
	$$
	 \left[\bar c_0+\sum_{\alpha=1}^2\bar c_\alpha\xi_\alpha-c\xi_1\xi_2\right]\partial_2f+\left[\xi_2-c\xi_1^2\right]\partial_1f+(\bar c_2-3c\xi_1)f=\left[c_0+\sum_{\alpha=1}^2c_\alpha\xi_\alpha\right]\partial_2f+c_2f=0.
	$$
Writing the above as $\partial_1f=F_1(\xi)f,\partial_2f=F_2(\xi)f$, a locally non-vanishing defined solution $f$ exists if and only if $\partial_1\log|f|=F_1(\xi),\partial_2\log|f|=F_2(\xi)$ has a solution. This amounts to the fact that the one-form $\theta=F_1{\rm d}\xi_1+F_2{\rm d}\xi_2$ is closed. In coordinates,
\begin{equation}\label{RelHam}
\theta=\frac{1}{\xi_2-c\xi_1^2}\left(3c\xi_1-\bar c_2+\frac{c_2(\bar c_0+\sum_{\alpha=1}^2\bar c_\alpha\xi_\alpha -c\xi_1\xi_2)}{c_0+\sum_{\alpha=1}^2c_\alpha\xi_\alpha}\right){\rm d}\xi_1-\frac{c_2}{c_0+\sum_{\alpha=1}^2c_\alpha\xi_\alpha}{\rm d}\xi_2.
\end{equation}
A long but straightforward computation shows that ${\rm d}\theta=0$ if and only if 
\begin{multline*}
\partial_2F_1-\partial_1F_2=c_0c_2\bar c_0-c_0^2\bar c_2+(3cc_1^2-c^2c_0c_2-cc_2^2\bar c_1+cc_1c_2\bar c_2)\xi_1^3+(2c_2^2\bar c_0-2c_0c_2\bar c_2)\xi_2\\+c_1c_2\xi_2^2+\xi_1^2(6cc_0c_1-cc_2^2\bar c_0+c_1c_2\bar c_1-c_1^2\bar c_2+cc_0c_2\bar c_2+4cc_1c_2\xi_2)\\+\xi_1(3cc_0^2+c_1c_2\bar c_0+c_0c_2\bar c_1-2c_0c_1\bar c_2)+(6cc_0c_2+2c_2^2\bar c_1-2c_1c_2\bar c_2)\xi_1\xi_2+2cc_2^2\xi_1\xi_2^2=0.
\end{multline*}
Obviously, this happens if ond only if $\Gamma_1=0$, which contradicts our assumption $\Gamma_1\neq 0$, i.e. the second-order Riccati chain equation is strictly non-autonomous, and $V_2$ cannot consist of Hamiltonian vector fields relative to any symplectic structure.
	
In the case of strictly non-autonomous affine second-order Riccati chain equations, the conditions to admit a Vessiot--Guldberg Lie algebra of Hamiltonian vector fields relative to a symplectic structure read as in (\ref{RelHam}) but with $c=0$. Hence,
\begin{multline*}
\partial_2F_1-\partial_1F_2=c_0(c_2\bar c_0-c_0\bar c_2)+2c_2(c_2\bar c_0-c_0\bar c_2)\xi_2+c_1c_2\xi_2^2+\xi_1^2c_1(c_2\bar c_1-c_1\bar c_2)\\+\xi_1(c_1(c_2\bar c_0-c_0\bar c_2)+c_0(c_2\bar c_1-c_1\bar c_2))+2c_2(c_2\bar c_1-c_1\bar c_2)\xi_1\xi_2=0.
\end{multline*}
A necessary condition for this equality to hold is $c_1c_2=0$. So, the above equation is equivalent to
\begin{multline*}
c_0(c_2\bar c_0-c_0\bar c_2)+2c_2(c_2\bar c_0-c_0\bar c_2)\xi_2+c_1c_2\xi_2^2-\xi_1^2c_1^2\bar c_2\\+\xi_1(c_1(-c_0\bar c_2)+c_0(c_2\bar c_1-c_1\bar c_2))+2c_2^2\bar c_1\xi_1\xi_2=0.
\end{multline*}
Another necessary conditions are $c_1\bar c_2=0, c_2\bar c_1=0$. Hence, the previous equation has the same solutions as
\begin{equation*}
c_0(c_2\bar c_0-c_0\bar c_2)+2c_2(c_2\bar c_0-c_0\bar c_2)\xi_2+c_1c_2\xi_2^2-\xi_1^2c_1^2\bar c_2+2c_2^2\bar c_1\xi_1\xi_2=0.
\end{equation*}
The above equation has the same set of solutions as the system
$$
c_1(c_2^2+\bar c_2^2)=0,\qquad (c_0^2+c_2^2)(c_2\bar c_0-c_0\bar c_2)=0,\qquad \bar c_1 c_2=0.
$$
Let us write down all solutions by analyzing the first equations. There are two options $c_1=0$ or $c_2=\bar c_2=0$. If $c_2=\bar c_2=0$, then the above system of equations is satisfied and we obtain the case a) detailed in the present lemma. If $c_1=0$, then $c_2\bar c_0-c_0\bar c_2=0$ and $\bar c_1c_2=0$. The last condition gives two subcases: $\bar c_1=0$ or $c_2=0$. The subcase $\bar c_1=0$ gives the case b) detailed in the present. Meanwhile, $c_2=0$ leads to $c_0\neq 0$ and $\bar c_2=0$. Nevertheless, this case is a particular subcase of case a). Hence, ${\rm d}\theta=0$ in the case $c=0$ if and only if some of the two sets of conditions (\ref{AffRiccHam}) are satisfied.
\end{proof}

\begin{note} Observe that the conditions (\ref{AffRiccHam}) do not depend on the chosen $\Gamma_0$ and $\Gamma_1\neq 0$.
\end{note}
\begin{theorem}\label{MT2} A strictly non-autonomous second-order Riccati chain equation can be mapped through an autonomous diffeomorphism into a  Euclidean Riccati equation if and only if it takes the form:
\begin{equation}\label{Class}
\dfrac{{\rm d}^2u}{{\rm d}x^2}=-3cu\dfrac{{\rm d}u}{{\rm d}x}-c^2u^3+f(x)c_0+c_0(c_1+\bar c_2)+\left[f(x)c_1+{c_1^2/2-1}\right]u+[f(x)c_2+\bar c_2]\left(cu^2+\dfrac{{\rm d}u}{{\rm d}x}\right). 
\end{equation}
for any non-constant $x$-dependent function $f(x)$, coefficients $c_1,c_0\in\mathbb{R}$ such that $c_1^2-4c_0c_2c<0$ with  $c_2\in \mathbb{R}\backslash\{0\}$, and arbitrary $\bar c_2\in \mathbb{R}$. A strictly non-autonomous second-order affine Riccati chain equation is diffeomorphic to a Euclidean Riccati equation if and only if 
$$
a)\,\,\alpha_1(x)=\alpha_2(x)=0, \,\,\bar c_2^2+4\bar c_1<0\qquad b)\,\,\alpha_0(x)=\alpha_1(x)=\bar c_0=\bar c_1=0.
$$
\end{theorem}
\begin{proof}

Let $V_2$ be the irreducible Lie algebra of $X^{\rm RC}_2$. If $X^{\rm RC}_2$ can be mapped into a Euclidean Riccati equation on $\mathbb{R}^2$, then $V_2$ must consist of Euclidean vector fields relative to some Euclidean metric on $\mathbb{R}^2$. The Lie algebra $V_2$ is generated by $\Gamma_0$, the vector fields $\{\Delta_{x,x_0}\}_{x\in\mathbb{R}}$, and their successive Lie brackets. Let us determine under which conditions $\Gamma_0$ and $\{\Delta_{x,x_0}\}_{x\in\mathbb{R}}$ generate a Lie algebra $V_2$ of Euclidean vector fields relative to a Euclidean metric on $\mathbb{R}^2$.

 In view of Lemma \ref{con}, every two linearly independent (over $\mathbb{R}$) Euclidean vector fields $X_1,X_2$  satisfy $X_1\wedge X_2\neq 0$. As the vector fields $\{\Delta_{x,x_0}\}_{x\in\mathbb{R}}$ are Euclidean by assumption, $\Delta_{x_1,x_0}\wedge\Delta_{x_2,x_0}=0$ for arbitrary $x_1,x_2\in \mathbb{R}$, and there is one non-zero $\Delta_{x,x_0}$ because the second-order Riccati chain equation is strictly non-autonomous, we see that the $\{\Delta_{x,x_0}\}_{x\in\mathbb{R}}$ must all generate a one-dimensional linear space. Hence, there exists an $x$-dependent function $f(x)$ such that $\Delta_{x,x_0}=f(x)\Gamma_1$ for $\Gamma_1\neq 0$.
Since $X^{\rm RC}_2$ is not autonomous, we get that $f(x)$ is non-constant and the $\{(X^{\rm RC}_2)_x\}_{x\in\mathbb{R}}$ generate, at least, a two-dimensional Lie algebra spanned by $\Gamma_0,\Gamma_1$ and their successive Lie brackets. Since $\Gamma_0\wedge \Gamma_1\neq 0$, the Lie algebra $V_2$ gives rise to a distribution of rank two.

Theorem \ref{ProjRicc} ensures that $\Gamma_0$ and $\Gamma_1$ are contained in the Vessiot--Guldberg Lie algebra P$_8\simeq \mathfrak{sl}(3)$. If we additionally  require $\Gamma_0$ and $\Gamma_1$ to be Euclidean, then Proposition \ref{ProjEuc}  states that $V_2$ must be diffeomorphic to a Lie subalgebra of P$_2$, P$_3$ or P$_4$. 
Let us study divide our analysis into  the case when $V_2$ is diffeomorphic to  Lie subalgebras of P$_4$ and when it is not.

$\diamond$ {\bf I) The Lie algebra $V_2$ is diffeomorphic to Lie subalgebras of {\rm P}$_4$}

If $V_2$ is diffeomorphic to a Lie subalgebra of P$_4\simeq \mathbb{R}^2\ltimes\mathbb{R}^2$, then any Lie bracket involving  elements of $V_2$, e.g. $\Gamma_0$ and $\Gamma_1$, must belong to an Abelian ideal of $V_2$. For instance, the vector fields $\Upsilon_1:=[\Gamma_1,\Gamma_0], \Upsilon_2:=[\Gamma_1,\Upsilon_1],\Upsilon_3:=[\Gamma_0,\Upsilon_1]$, with coordinate expressions
$$
\begin{aligned}
\Upsilon_1&=(c_0+c_{1}\xi_{1}+c_2\xi_2)\partial_ {1}+[c_0\bar{c}_2-c_2\bar{c}_0+(-cc_0+c_1\bar c_2-c_2\bar{c}_{1})\xi_{1}-c_{1}\xi_2]{\partial_2}\neq0,\\
\Upsilon_2&=c_2(c_0+c_{1}\xi_{1}+c_2\xi_{2}){\partial_{1}}\!-\![c_0(2c_{1}+c_2\bar c_2)\!-\!c_2^2\bar c_0+(2c_{1}^2\!-\!cc_0c_2\!-\!c_2^2\bar c_{1}\!+\!c_{1}c_2\bar c_2)\xi_{1}+c_{1}c_2\xi_2]\partial_2,\\
\Upsilon_3&=[2c_2\bar{c}_0-\bar{c}_{2}c_0+cc_1\xi_1^2+(3cc_0+2c_2\bar c_1-c_{1}\bar c_{2}+cc_2\xi_2)\xi_{1}+(2c_{1}+c_2\bar c_2)\xi_2]\partial_1+\\
&[-c_{1}\bar c_0+c_2\bar c_0\bar c_2-c_0(\bar c_{1}+\bar c_2^2)+(-2\bar c_{1} c_{2}+c_1\bar c_{2}+cc_2\xi_2)\xi_2\\&+(-cc_2\bar c_0-2c_1\bar c_1+2cc_{0}\bar c_2+c_2\bar c_1\bar c_2-c_1\bar c_2^2+cc_1\xi_2)\xi_{1}]\partial_2
\end{aligned}
$$
must satisfy the relations $[\Upsilon_1,\Upsilon_2]=[\Upsilon_1,\Upsilon_3]=[\Upsilon_2,\Upsilon_3]=0$ giving rise to restrictions on the form of $\Gamma_0$ and $\Gamma_1$.  For example, we have to impose 
\begin{multline*}
0=[\Upsilon_1,\Upsilon_2]=\left[2 c_{2}(c_0 c_{1}+c_{0} c_2\bar c_{2}-c_{2}^2 \bar c_{0})+2c_2(c_{1}^2 -cc_0c_2+c_{1} c_{2} \bar c_{2} -c_{2}^2\bar c_1)\xi_{1}\right]\partial_{1}\\
[-2(2c_{0}c_1^2-cc_0^2c_2-c_1c_2^2\bar c_0- c_0c_2^2\bar c_1+2c_{1}c_0c_2\bar c_2) \\-4c_1(c^2_{1} -cc_0c_2-c_2^2\bar c_{1}+c_1\bar c_2c_2) \xi_{1}-2 c_2(c^2_{1}-cc_0c_2-c_2^2\bar c_{1}+c_{1}c_2\bar c_2) \xi_{2}]{\partial_{2}}.
\end{multline*}
The above vanishes, along with the Lie brackets $[\Upsilon_ 1,\Upsilon_3],[\Upsilon_2,\Upsilon_3]$, if and only if 
\begin{equation}\label{CoeCon}
a)\,\,c_2\neq 0, \,\,\bar c_0=\frac{c_0}{c^2_2}(c_{1}+c_2\bar c_2),\,\,\bar c_{1}=\frac{c_1^2-cc_0c_2+c_{1}c_2\bar c_2}{c^2_2},\qquad b)\,\,c=c_1=c_2=0.
\end{equation}

{\bf Case I.a:} Since $c_2\neq 0$, the vector field $\Gamma_1$ can be rescaled and $c_2$ can be assumed to be equal to one without varying $V_2$. A set of generators for the Lie algebra $V_2$  reads:
\begin{equation}\label{Bas}
\begin{gathered}
Y_1:=(c_0+c_{1}\xi_{1}+\xi_{2})\partial_2,\quad Y_3:=(c_0+c_{1}\xi_{1}+\xi_2)[\partial_{1}-c_{1}\partial_2],\\
Y_2:=(-c\xi_1^2+\xi_2)\partial_{1}+\{c_0(c_{1}+\bar c_2)+\bar c_2\xi_2+\xi_1[-cc_0+c_1(c_1+\bar c_{2})-c\xi_{2}]\}\partial_2, \\ Y_4:=\{c_0(2c_1+\bar c_2)+cc_1\xi_1^2+(2c_1+\bar c_2)\xi_2+\xi_1[cc_0+c_1(2c_1+\bar c_2)+c\xi_2]\}\partial_1\\+\{cc_0^2-c_0c_1(2c_1+\bar c_2)-[-2cc_0+c_1(2c_1+\bar c_2)]\xi_2+c\xi_2^2+\xi_1[cc_0c_1-c_1^2(2c_1+\bar c_2)+cc_1\xi_2]\}\partial_2.
\end{gathered}
\end{equation}
Indeed, the commutation relations between $Y_1,Y_2,Y_3,Y_4$ read:
\begin{equation}
\begin{gathered}\label{ComY}
\left[Y_1,Y_2\right]=Y_3,\qquad  [Y_1,Y_3]=Y_3,\qquad [Y_2,Y_3]=Y_4,\qquad [Y_1,Y_4]=Y_4,\qquad [Y_3,Y_4]=0,\\
[Y_2,Y_4]=(3c_1+2\bar c_2)Y_4-(cc_0+(c_1+\bar c_2)(2c_1+\bar c_2))Y_3.
\end{gathered}
\end{equation}
Let us assume the affine, $c=0$, and non-affine, $c\neq 0$, subcases.

{\it I.a.1) Non-affine subcase:} The vector fields $Y_1,\ldots, Y_4$ become a basis for $V_2$. Indeed, if $\sum_{\alpha=1}^4\mu_\alpha Y_\alpha=0$, then
$$
\begin{gathered}
\left[\partial_1,\left[\partial_2,\sum_{\alpha=1}^4\mu_\alpha Y_\alpha\right]\right]=\mu_4c(\partial_1+c_1\partial_2)-c\mu_2\partial_1=0, \quad\left[\partial_2,\left[\partial_2,\sum_{\alpha=1}^3\mu_\alpha Y_\alpha\right]\right]=2c\mu_4\partial_2=0.
\end{gathered}
$$
Hence, $\mu_2=\mu_4=0$ and, from here, it follows that $\mu_1Y_1+\mu_3Y_3=0$ implies that $\mu_1=\mu_3=0$.
Therefore, the Lie algebra $V_2$ has an Abelian two-dimensional ideal given by $\langle Y_3,Y_4\rangle$ and 
$$
Y_3\wedge Y_4=c(c_0+c_1\xi_1+\xi_2)^3\partial_1\wedge\partial_2
\neq 0.$$ 
The elements $Y_1,Y_2-Y_3$ act on the ideal $\langle Y_3,Y_4\rangle$ according to the matrices in the basis $\{Y_3,Y_4\}$ given by
\begin{equation}\label{action}
[{\rm ad}_{Y_1}]=\left(\begin{array}{cc}
1&0\\
0&1
\end{array}\right),\qquad
[{\rm ad}_{Y_2-Y_3}]=\left(\begin{array}{cc}
0& -(cc_0+3c_1\bar c_2+\bar c_2^2+2c_1^2)\\
1&3 c_1+2\bar c_2
\end{array}\right).
\end{equation}
The elements of $V_2$ act irreducibly on $\langle Y_3,Y_4\rangle$ if and only if 
$$
0>(3c_1+2\bar c_2)^2-4(cc_0+3c_1\bar c_2+\bar c_2^2+2c_1^2)=c_1^2-4c_0c.
$$
In view of Lemma \ref{8Lie}, $V_2$ is  diffeomorphic to P$_4$ under the assumed conditions. It is only left to recall that the above condition applies to a rescaled $\Gamma_1$ with $c_2=1$. Hence, the condition for a general $\Gamma_1$ with no rescaled coefficients reads $c_1^2-4c_0c_2c<0$. Indeed, recall that the fact $X^{\rm RC}_2$ is associated with a conformal Riccati equation is independent of the chosen $\Gamma_1$ and that this vector field can be determined up to a proportional constant without varying $V_2$.

{\it I.a.2) Affine subcase}: We can choose among $\Gamma_0,\Gamma_1$ and the vector fields of (\ref{Bas}) a set of generators of $V_2$  of the form
\begin{equation}\label{Gen4}
Y_1:=(c_0+c_{1}\xi_{1}+\xi_{2})\partial_2,\quad Y_3:=(c_0+c_{1}\xi_{1}+\xi_2)\partial_{1},\quad
Y_2:=\xi_2(\partial_{1}-c_1\partial_2). 
\end{equation}
Since $(Y_3-c_1Y_1)\wedge Y_2=0$, Lemma \ref{con} ensures that $V_2$ is not a Lie algebra of Euclidean vector fields unless $Y_3-c_1Y_1$ and $Y_2$ are linearly dependent. This only happens for $c_0=c_1=0$, which in view of (\ref{CoeCon}) implies that $\bar c_0=\bar c_1=0$. 
In this case, $V_2$ is non-Abelian, two-dimensional and it spans a distribution of rank two. Hence, it becomes diffeomorphic to a Lie subalgebra of P$_4$.


{\bf Case I.b:} If $c=c_1=c_2=0$, then the condition $\Gamma_1\neq 0$ allows us to assume, by rescaling $\Gamma_1$, that $c_0=1$ without changing $V_2$. Hence, $V_2$ possesses a basis
\begin{equation}\label{Bas2}
\begin{gathered}
Y_1:=\partial_2,\quad Y_2:=\xi_2\partial_{1}+[\bar c_0+\bar c_1\xi_1+\bar c_2\xi_2]\partial_2,\quad
Y_3:=\partial_1,
\end{gathered}
\end{equation}
which admits an invariant distribution generated by $2\partial_1+(\bar c_2\pm \sqrt{\bar c_2^2+4\bar c_1})\partial_2$ for $\bar c_2^2+4\bar c_1\geq 0$ and it is primitive otherwise. Proposition \ref{ConLie} implies that the only three-dimensional  imprimitive Euclidean Lie algebra of vector fields is I$^{\alpha=1}_8$. Nevertheless, $\dim [{\rm I}^{\alpha=1}_8, {\rm I}^{\alpha=1}_8]=2$ and every element of  $[{\rm I}^{\alpha=1}_8, {\rm I}^{\alpha=1}_8]$ is an ideal of I$^{\alpha=1}_8$. As $[V_2,V_2]= \langle\partial_1+\bar c_2\partial_2,\bar c_1\partial_2\rangle$, the Lie algebra $V_2$ can only be diffeomorphic to I$_8^{\alpha=1}$ for $\bar c_1\neq 0$. But in this case $Y_1\in [V_2,V_2]$ and  $Y_1$ does not span an ideal of I$^{\alpha=1}_8$. Hence, $V_2$ is not diffeomorphic to I$^{\alpha=1}_8$. In view of Table \ref{table3}, the only three-dimensional primitive Lie algebra of vector fields on $\mathbb{R}^2$ is P$_1$. Then, $V_2$ is diffeomorphic to P$_1$ for $\bar c_2^2+4\bar c_1<0$.

$\Diamond$ {\bf II) The Lie algebra $V_2$ is not diffeomorphic to Lie subalgebras of {\rm P}$_4$}

Let us assume that  $\Gamma_0$ and $\Gamma_1$ generate a Lie algebra $V_2$ of Euclidean projective vector fields that is not diffeomorphic to a Lie subalgebra of P$_4$. In view of Proposition \ref{ProjEuc} and Table \ref{table5},  the Lie algebra $V_2$ must be diffeomorphic to a Lie subalgebra of P$_2$ or P$_3$, which are Lie algebras of Hamiltonian vector fields (cf. \cite{BBHLS15}). Lemma \ref{RiccNoHam} ensures that $c=0$ and $\Gamma_0$ and $\Gamma_1$ satisfy the conditions (\ref{AffRiccHam}). We have two cases:


{\bf Case II.a:} A first option is $c_2=\bar c_2=0$ . Hence, $V_2$ is generated by  
\begin{equation}\label{Gen3}
\begin{gathered}
Y_0:=\Gamma_0=\xi_2\partial_1+(\bar c_0+\bar c_1\xi_1)\partial_2, \qquad Y_1:=\Gamma_1=(c_0+c_1\xi_1)\partial_2,\\
Y_2:=(c_0+c_1\xi_1)\partial_1-c_1\xi_2\partial_2, \qquad Y_3:=(c_0\bar c_1-c_1\bar c_0)\partial_2.
\end{gathered}
\end{equation}
The above vector fields span a Lie algebra of dimension bigger than three for $c_1(c_0\bar c_1-c_1\bar c_0)\neq 0$. Therefore, $V_2$ can be a Lie subalgebra of P$_2$ or P$_3$ provided $c_1(c_0\bar c_1-c_1\bar c_0)=0$. If $c_1=0$, then $[Y_1,Y_2]=0$ and $V_2$  cannot be isomorphic to a Lie subalgebra of P$_2$ or P$_3$. Meanwhile, assuming that $c_0\bar c_1-c_1\bar c_0=0$, with $c_1\neq 0$, leads to the following commutation relations
$$
[Y_0,Y_1]=-Y_2,\qquad [Y_1,Y_2]=-2c_1Y_1,\qquad [Y_0,Y_2]=2c_1Y_0-4\bar c_1Y_1.
$$
By making use of the Killing form related to this Lie algebra, we find that it is non-degenerate of signature $(2,1)$. Hence, $V_2$ is isomorphic to $\mathfrak{sl}(2)$. The Casimir element related to it is, up to a proportional constant, $Y_1\otimes Y_0+Y_0\otimes Y_1-\frac{2\bar c_1}{c_1}Y_1\otimes Y_1+\frac{1}{2c_1}Y_2\otimes Y_2$. Its determinant is zero. In view of Theorem 4.4 in \cite{BBHLS15}, the Lie algebra $V_2$ is diffeomorphic to I$_5$ and it is not diffeomorphic to P$_2$.

{\bf Case II.b:} The following case is given by $c_1=\bar c_1=0$ where $(c_2,c_0)$ and $(\bar c_2,\bar c_0)$ are linearly dependent. The Lie algebra is spanned by vector fields of the form 
\begin{equation}\label{Gen}
Y_0:=\xi_2\partial_1,\qquad Y_1:=(c_0+c_2\xi_2)\partial_2,\qquad Y_2:=(c_0+c_2\xi_2)\partial_1.
\end{equation}
If these vector fields span a three-dimensional Lie algebra, i.e. $c_0\neq 0$, then they admit a two-dimensional Abelian ideal $\langle Y_1,Y_2\rangle$ and $V_2$ is neither not isomorphic to P$_2$ nor to P$_3$. If $\dim V_2=2$, namely $c_0=0$, then $V_2$ may be diffeomorphic to a two-dimensional Lie algebra of type I$^{r=1}_{14A}$ that is diffeomorphic to a Lie subalgebra of P$_4$ and already appeared in the previous subsection. 

\end{proof}

Finally, we characterize strictly non-autonomous second-order Riccati chain equations that can be related to hyperbolic Riccati equations. 

\begin{theorem}\label{MT3} A strictly non-autonomous second-order Riccati chain equation is diffeomorphic, as a first-order system, to a hyperbolic Riccati equation when it takes the form
\begin{equation}\label{Class2}
\dfrac{{\rm d}^2u}{{\rm d}x^2}=-3cu\dfrac{{\rm d}u}{{\rm d}x}-c^2u^3+f(x)c_0+c_0(c_1+\bar c_2)+\left[f(x)c_1+{c_1^2/2-1}\right]u+[f(x)c_2+\bar c_2]\left(cu^2+\dfrac{{\rm d}u}{{\rm d}x}\right). 
\end{equation}
where $f(x)$ is a non-constant function, $c_1^2-4c_0c_2c>0$ with $c,c_2\neq 0$, and $\bar c_2\in \mathbb{R}$. A strictly non-autonomous affine second-order Riccati chain equation is diffeomorphic to a hyperbolic Riccati equation if and only if it takes the form (\ref{Class2})  for $c=0$ and $c_2\neq 0$ or it satisfies the following conditions
$$
\begin{gathered}
a)\,\,\alpha_1(x)=\alpha_2(x)=0, \,\,\bar c_2^2+4\bar c_1>0,\qquad b)\,\,\alpha_0(x)=\alpha_1(x)=\bar c_0=\bar c_1=0,\\
c)\,\, c_1=\bar c_1=0,c_1\neq 0, c_0\bar c_2-c_2\bar c_0=0.
\end{gathered}
$$
\end{theorem}
\begin{proof}
If $X^{\rm RC}_2$ is diffeomorphic to a hyperbolic Riccati equation, then its irreducible Lie algebra consists of projective and hyperbolic vector fields. Hence, Proposition \ref{HypPro} states that $X^{\rm RC}_2$ admits an irreducible Vessiot--Guldberg Lie algebra diffeomorphic to a Lie subalgebra of I$_4$ or I$_9$. As a consequence, further analysis is divided into two cases: $V_2$ diffeomorphic to a Lie subalgebra of ${\rm I}_9$, and $V_2$ diffeomorphic to a Lie subalgebra of ${\rm I}_4$ not described in the previous case.

$\Diamond$ {\bf I) The Lie algebra $V_2$ is diffeomorphic to a Lie subalgebra of {\rm I}$_9$}

 If  $V_2$ is diffeomorphic to a Lie subalgebra of I$_9$, then $\Upsilon_1:=[\Gamma_1,\Gamma_0]$, $\Upsilon_2:=[\Gamma_1,\Upsilon_1]$, $\Upsilon_3:=[\Gamma_0,\Upsilon_1]$ belong to $[{\rm I}_9,{\rm I}_9]\simeq \mathbb{R}^2$ and commute among themselves, i.e. $[\Upsilon_1,\Upsilon_2]=[\Upsilon_1,\Upsilon_3]=[\Upsilon_2,\Upsilon_3]=0$. This establishes conditions on the coefficients of $\Gamma_0$ and $\Gamma_1$, namely the conditions (\ref{CoeCon}) found in Theorem \ref{MT2}. We investigate the subcases $c=0$ and $c\neq 0$.

{\it I.a) Non-affine subcase:} If we assume $c\neq 0$, then $\Gamma_1$ and $\Gamma_0$ must satisfy condition a) in (\ref{CoeCon}). It was already showed in the subcase I.a.1) of the proof of Theorem \ref{MT2} that the Lie algebra $V_2$ related to this case has a basis $Y_1,\ldots,Y_4$ given by (\ref{Bas}). Since $\dim {\rm I}_9=4$ and $V_2$ is assumed to be diffeomorphic to a Lie subalgebra of ${\rm I}_9$, it follows that $V_2$ is diffeomorphic to I$_9\simeq\mathfrak{h}_2\oplus\mathfrak{h}_2$. It was also proved in the subcase I.a.1) of the proof of Theorem \ref{MT2} that $V_2$ is diffeomorphic to P$_4$ for $c_1^2-4c_0c_2c<0$. Following the process given there, we obtain that if $c_1^2-4c_0c_2c>0$, then $[V_2,V_2]$ can be written as a non-trivial direct sum of subspaces invariant under the adjoint action of $V_2$. In view of Table \ref{table3} and recalling that the Lie algebra $[V_2,V_2]=\langle Y_3,Y_4\rangle$ is two-dimensional and it spans a distribution of rank two,  $V_2$ becomes diffeomorphic to I$_9$. If $c_1^2-4c_0c_2c=0$, then not every element of $V_2$ diagonalizes when acting on $[V_2,V_2]$ as it happens for I$_9$ acting on $[{\rm I}_9,{\rm I}_9]$. So, $V_2$ is not diffeomorphic to I$_9$.
 
 
{\it I.b) Affine subcase:} Let us consider both subcases given in conditions (\ref{CoeCon}) for $c=0$.
Consider the first set of conditions in (\ref{CoeCon}). It was proved in subcase I.a.2) of the proof of Theorem \ref{MT2} that $V_2$ admits a set of generators $Y_1,Y_2,Y_3$ of the form (\ref{Gen4}). In fact,
$$[Y_1,Y_2]=Y_3-c_1Y_1,\qquad [Y_2,Y_3]=0,\qquad [Y_1,Y_3]=Y_3-c_1Y_1.
$$ 
In view of (\ref{Gen4}),  $\dim V_2=3$ if and only if $c_0^2+c_1^2\neq0$. Assume that $\dim V_2=3$. In this case, $Y_2-Y_3$ belongs to the center of the Lie algebra and $(Y_3-c_1Y_1)\wedge Y_2=0$. In view of Table \ref{table3} and Table \ref{table5},  the Lie algebra $V_2\simeq \langle Y_2-Y_3\rangle \oplus \langle Y_3-c_1Y_1,V_2\rangle$ is diffeomorphic to I$^{r=1}_{15B}$ and $V_2$ becomes diffeomorphic to a hyperbolic Lie algebra on $\mathbb{R}^2$. 
If $c_0=c_1=0$, then $V_2$ is non-Abelian,  two-dimensional and spans a distribution of rank two. In view of Table \ref{table3}, the Lie algebra $V_2$ becomes diffeomorphic to a Lie algebra of hyperbolic projective vector fields: ${\rm I}^{r=1}_{14A}$, which is not a Lie subalgebra. Due to the relations (\ref{CoeCon}), this case leads to $\bar c_0=\bar c_1=0$. This shows that second-order Riccati equations (\ref{Class2}) with $c_2\neq 0$ and $c=0$ are locally diffeomorphic to a hyperbolic Riccati equation. 

Let us assume the second set of conditions in (\ref{CoeCon}), namely $c_2=c_1=c=0$. As shown in subcase I.b) of Theorem \ref{MT2}, the Lie algebra $V_2$ is a three-dimensional Lie algebra with a basis  (\ref{Bas2}) and $V_2$ will be imprimitive provided that $\bar c_2^2+4\bar c_1\geq 0$. Lemma \ref{one} ensures that hyperbolic Lie algebras are imprimitive, so condition $\bar c_2^2+4\bar c_1\geq 0$ must be satisfied. Observe that $V_2$ has an ideal $[V_2,V_2]=\langle Y_1,Y_3\rangle$ and every element of $V_2$ diagonalize when acting on it for $\bar c_2^2+4\bar c_1>0$. Additionally two elements of $V_2$ are proportional at each point if and only if $\bar c_1=0$. In view of Lemma \ref{8Lie}, the Lie algebra $V_2$ is diffeomorphic to I$_8$ for $\bar c_1\neq 0$.  If $\bar c_1=0$, then $V_2$ is diffeomorphic to ${\rm I}_{15B}^{r=1}$.

$\Diamond$ {\bf II) The Lie algebra $V_2$ is diffeomorphic to a Lie subalgebra of {\rm I}$_4$ and it is not diffeomorphic to a Lie subalgebra of {\rm I}$_9$}

Since I$_4$ is three-dimensional. Hence, $V_2$ is a three or two-dimensional Lie algebra, which in view of Table \ref{table5} implies that $V_2$ must be a Lie algebra of Hamiltonian vector fields with $c=0$ satisfying conditions (\ref{AffRiccHam}): 
\begin{itemize}
\item Subcase $c_2=\bar c_2=0$.  It was proved in the case II.a of the proof of Theorem \ref{MT2} that $V_2$ is spanned  by the vector fields (\ref{Gen3}). So, $\dim V_2<4$ if and only if $\bar c_1(c_0\bar c_1-c_1\bar c_0)=0$. If $c_1=0$, then $[Y_1,Y_2]=0$ and $V_2$ cannot be isomorphic to I$_4\simeq\mathfrak{sl}(2)$. If $c_1\neq 0$ and $c_0\bar c_1-c_1\bar c_0=0$, then $V_2$ becomes isomorphic to I$_5$ (see again case II.a in Theorem \ref{MT2}.
\item Subcase $c_1=\bar c_1=0$ with $c_0\bar c_2-c_2\bar c_0=0$. It was proved in case II.b of the proof of Theorem \ref{MT2} that $V_2$ is spanned by the vector fields (\ref{Gen}). If $c_0\neq0$, then $V_2$ is three-dimensional and $Y_1,Y_2$ commute. Hence, $V_2$ is not isomorphic to I$_4$. If $c_0=0$, then $\bar c_0=0$ also and $V_2$ is diffeomorphic to I$^{r=1}_{14A}$ that was already described in part I.a.2 of this proposition. 
\end{itemize}

\end{proof}

\section{Applications of the Riccati hierarchy to partial differential equations}
In this section we discuss the second-order Riccati chain equation or, equivalently, an associative projective Riccati equation in order to study B\"acklund transformations for the Sawada--Kotera (SK) and Kaup-Kupershmidt (KK) equations. 

Let $u$ be a real function on $\mathbb{R}^2$. The Sawada--Kotera  \cite{SK74} and Kaup--Kupershmidt \cite{KK80,FG80}  equations take the form
$$
\begin{gathered}
u_t+(u_{4x}+30uu_{xx}+60u^3)_x=0,\\
u_t+\left(u_{4x}+30uu_{xx}+\frac{45}{2}u_x^2+60u^3\right)_x=0,
\end{gathered}
$$
respectively. Both partial differential equations are related to the linear spectral problem \cite{KK80}
\begin{equation}\label{lsp}
\Psi_{xxx}+6u\Psi_x+(6R-\lambda)\Psi=0,
\end{equation}
where $\Psi,R:\mathbb{R}\rightarrow \mathbb{R}$ are $x$-dependent functions and $\lambda$ is a spectral parameter. More specifically, the linear spectral problem (\ref{lsp}) for the SK equations  has $R=0$ and $R=u_x/2$ for the KK equations. The linear spectral problem gives rise to the Darboux transformations \cite{LR88}
\begin{equation}\label{Back}
({\rm SK})\qquad \bar u-u=\partial_x^2\log \Psi,\qquad\qquad ({\rm KK})\qquad \bar u-u=\frac 12\partial_x^2\log(\Psi\Psi_{xx}-\frac 12 \Psi_x^2+3u\Psi^2).
\end{equation}

Let us prove that Riccati equations of projective and conformal type can be used to study the above B\"acklund transformations. A third-order differential equation can be considered as a first-order system on the second-order tangent bundle, ${\rm T}^2\mathbb{R}\simeq\mathbb{R}^3$, by adding new variables $v:={\rm d}\Psi/{\rm d}x$ and $a:={\rm d}v/{\rm d}x$. Therefore, (\ref{lsp}) can be studied through the linear system
\begin{equation}\label{LSPfirst}
\left\{\begin{aligned}
\frac{{\rm d}\Psi}{{\rm d}x}&=v\\
\frac{{\rm d}v}{{\rm d}x}&=a\\
\frac{{\rm d}a}{{\rm d}x}&=-6uv+(\lambda-6R)\Psi\\
\end{aligned}\right.\Longleftrightarrow \qquad \frac{{\rm d}}{{\rm d}x}\left[\begin{array}{c}\Psi\\v\\a\end{array}\right]=
\left[\begin{array}{ccc}0&1&0\\0&0&1\\\lambda-6R&-6u&0\\
\end{array}\right]\left[\begin{array}{c}\Psi\\v\\a\end{array}\right].
\end{equation}
This is a linear system of differential equations associated with the Vessiot--Guldberg Lie algebra spanned by the vector fields $X_{ij}=x_i\partial/\partial x_j$ with $i\neq j$. The linear function $\Psi:V\rightarrow \mathfrak{sl}(3)$ mapping each vector field $X_{ij}$ into the traceless $n\times n$ matrix $M_{ij}$ with coefficients $(M_{ij})_k^l:=-\delta_k^i\delta^l_j$ is a Lie algebra isomorphism.

 If we set $\mathcal{O}:=\{(\Psi,v,a)\in {\rm T}^2\mathbb{R}\simeq\mathbb{R}^3:\Psi\neq 0\}$, we can define the projection $\pi:\mathcal{O}\ni (\Psi,v,a)\mapsto (y_1,y_2):=(v/\Psi,a/\Psi)\in \mathbb{R}^2$. Since the system (\ref{LSPfirst}) is linear and every solution can be multiplied by a constant to get another solution, we can perform a reduction which is equivalent to applying the projection $\pi$. Indeed, all elements of $V$ are projectable onto $\mathbb{R}^2$. The kernel of the Lie algebra morphism $\pi_*:X\in V\rightarrow \pi_*X\in \pi_*V$ is an ideal of $V$. Since $V\simeq\mathfrak{sl}(3)$ is simple, it has only the trivial ideals $0$ and $V$. But if $\pi_*$ is not identically zero, then $\ker \pi_*=\{0\}$ and $\pi_*V\simeq V$. Since the linear system (\ref{LSPfirst}) is associated with an $x$-dependent vector field taking values in $V$, its projection is determined by an $x$-dependent vector field taking values in the Vessiot--Guldberg Lie algebra $\pi_*V\simeq\mathfrak{sl}(3)$. It is known that every Lie algebra of vector fields on the plane isomorphic to $\mathfrak{sl}(3)$ is diffeomorphic to P$_8$ and an $x$-dependent vector field taking values in P$_8$ gives rise, up to a change of variables, to a projective Riccati equation (cf. \cite{BBHLS15}). In our case, the projection of (\ref{LSPfirst}) consists exactly of the following projective Riccati equation
$$
\left\{\begin{aligned}
\frac{{\rm d}y_1}{{\rm d}x}&=y_2-y_1^2,\\
\frac{{\rm d}y_2}{{\rm d}x}&=-6uy_1+(6-\lambda R)-y_1y_2,\\
\end{aligned}\right.
$$
which is related to a Vessiot--Guldberg Lie algebra of vector fields P$_8\simeq\mathfrak{sl}(3)$. Recall that Theorem \ref{MT1} tells us that this system is indeed equivalent to a second-order Riccati chain equation.

It is interesting that the B\"acklund transformations  (\ref{Back}) can be recast in the form
$$
\bar u=u+y_1-y_2^2,\qquad \bar u=u+\frac 12\partial_x\left[\frac{y_2+3\dot u+6 u y_1}{y_2-y_1^2/2+3u}\right].
$$
This shows that B\"acklund transformations for the KK and KS equations do not really depend on the linear spectral problem, but rather on the associated Riccati projective equations which contain all the necessary information for their description. 

Although the above procedure has been applied to the SK and KK equations, most of the above arguments can be applied to many other PDEs, such as the Boussinesq equation \mbox{
\cite{F83}
}
or the  Fitzhugh-Nagumo equations \cite{Ab08}, giving rise to similar results.

\section{Superposition rules for Gambier equations}
In this section we show how conformal Riccati equations can help in studying different types of Gambier equations. 

The second-order differential equation
\begin{equation}\label{G5}
\frac{{\rm d}^2y}{{\rm d}x^2}-\frac3{4y}\left(\frac{{\rm d}y}{{\rm d}x}\right)^2+\frac 32y^2\frac{{\rm d}y}{{\rm d}x}+\frac 14y^3+6uy-2\lambda=0,
\end{equation}
where $\lambda\in\mathbb{R}$ and $u$ is an arbitrary $x$-dependent function, belongs to the class of Gambier differential equation G25  \cite{Gambier}. For an arbitrary $x$-dependent function $u(x)$, this is not a Lie system when written as a first-order system by adding a new variable $v:={\rm d}y/{\rm d}x$. Indeed, consider the $x$-dependent vector field related to such a system
$$
X=v\partial_y+[3v^2/(4y)-3y^2v/2+y^3/4-2\lambda)\partial_v+6uy\partial_v.
$$
When $u(x)$ is not a constant function, the irreducible Lie algebra related to $X$ is spanned by the vector fields 
$$
X_1=v\partial_y+[3v^2/(4y)-3y^2v/2+y^3/4-2\lambda]\partial_v,\qquad X_2=y\partial_v,
$$
and their successive Lie brackets. It is a long but straightforward computation to show that $X_3:=[X_1,X_2] $ allows us to generate six vector fields $X_{k+1}:=[X_1,X_{k}]$, with $k=3,\ldots,8$. These vector fields are linearly independent over $\mathbb{R}$. It can be proved that they generate an infinite-dimensional Lie algebra of vector fields.

Nevertheless, the contact transformation \cite{Gambier}
$$
y:=\frac{\lambda}{{\rm d}z/{\rm d}x+z^2/2+3u}
$$
maps (\ref{G5}) into
$$
\frac{{\rm d}^2z}{{\rm d}x^2}+3z\frac{{\rm d}z}{{\rm d}x}+z^3+6uz+3\frac{{\rm d}u}{{\rm d}x}-\lambda=0,
$$
namely a second-order Riccati chain equation. This implies that every particular solution of G25 can be described through particular solutions of different types of second-order Riccati chain equations. We already proved that second-order Riccati equations are Lie systems when written as first-order systems. Hence, (\ref{G5}) allows us to use a Lie system in order to study a non-Lie system.

Lie systems can also be employed to study G27 \cite{Gambier}, namely
\begin{equation}\label{cc}
\frac{{\rm d}^2y}{{\rm d}x^2}=\frac{1}{2y}\left(\frac{{\rm d}y}{{\rm d}x}\right)^2-2cy\left(\frac{{\rm d}y}{{\rm d}x}\right)-c^2\frac{y^3}{2}-a(x)y-\frac{1}{2y}.
\end{equation}
The differential equation (\ref{cc}) can be written as a first-order system
\begin{equation}\label{G27First}
\left\{\begin{aligned}
\frac{{\rm d}y}{{\rm d}x}&=v,\\
\frac{{\rm d}v}{{\rm d}x}&=\frac{v^2}{2y}-2cyv-c^2\frac{y^3}{2}-a(x)y-\frac{1}{2y},
\end{aligned}\right.
\end{equation}
by adding a new variable $v:={\rm d}y/{\rm d}x$. Let us perform a change of variables given by
$$
y=y_1^{-1},\qquad v=-\frac{c+y_1y_2}{y_1^2}.
$$
This maps the Gambier equation (\ref{cc}), written as a first-order system, into a Riccati conformal equation of the form
$$\left\{
\begin{aligned}
\frac{{\rm d}y_1}{{\rm d}x}&=c+y_1y_2,\\
\frac{{\rm d}y_2}{{\rm d}x}&=a(x)+\frac 12(y_1^2+y_2^2).
\end{aligned}\right.
$$
Indeed, it is easy to prove that the vector field $X_x=(c+y_1y_2)\partial_{y_1}+[a(x)+(y_1^2+y_2^2)/2]\partial_{y_2}$ is, for every fixed $x\in \mathbb{R}$, a conformal vector field relative to the metric $dy_1\otimes dy_1-dy_2\otimes dy_2$. In this case, we have used an $x$-independent change of variables to map G27 into a Lie system, which shows that (\ref{G27First}) is a Lie system.
\section{Lax pair associated with the Sturm-Liouville problem}

Let us consider the Sturm-Liouville problem (SLP) for the function $w(x,\lambda)$ with a given potential function $u(x)$
\begin{equation}
\dfrac{{\rm d}^2}{{\rm d}x^2}w(x,\lambda)-\alpha\left(u(x),\lambda\right)w(x,\lambda)=0,\hspace{5mm}\lambda\in\mathbb{C}.
\label{eq1}
\end{equation}
The Sturm-Liouville problem appears in the analysis of relevant integrable systems, e.g. for the Darboux-Treibich-Verdier potentials \cite{Ve11}, which are the following rational trigonometric and elliptic potentials respectivaly:
$$
\alpha(z)=\frac{\alpha_1}{z^2}+\alpha_0,\quad \alpha=\frac{\alpha_1^2a^2}{\sin^2(az)}+\frac{\alpha_2a^2}{\cos^2(az)}+\alpha_0,\quad \alpha(z)=\sum_{i=1}^3\alpha_i\mathcal{P}(z-\omega_i)+\alpha_4+\alpha_4 \mathcal{P}(z)+\alpha_0.
$$
Here, $\alpha_1,\ldots,\alpha_4,a$ are arbitrary real constants and $\mathcal{P}(z)$ is the Weierstrass elliptic functions with periods $2\omega_1,2\omega_2$ and $\omega_3=\omega_1+\omega_2$.
The matrix linear problem for the wavefunction $\Phi\in SL(2,\mathbb{C})$ associated with (\ref{eq1}) has the form
\begin{equation}
\partial_x\Phi=L([u],\lambda)\Phi,\hspace{5mm}\mbox{where}\hspace{5mm}L([u],\lambda)=\begin{pmatrix} 0 & 1 \\ \alpha\left(u(x),\lambda\right) & 0 \end{pmatrix}\in \mathfrak{sl}(2,\mathbb{C}).
\label{eq2}
\end{equation}
We look for an $\mathfrak{sl}(2,\mathbb{C})$-valued matrix $M([u],\lambda)$ such that the Lax pair
\begin{equation}
\partial_xM+[M,L]=0
\label{eq3}
\end{equation}
holds and is equivalent to equation (\ref{eq1}). Here, we use the abbreviated notation of the jet space $[u]=(x,u,u_x,u_{xx},\ldots)$. The Lax pair (\ref{eq3}) can be regarded as the compatibility conditions of a linear spectral problem (LSP) of the form \cite{GP12}
\begin{equation}
 \partial_x\Phi([u],\lambda,y)=L([u],\lambda)\Phi([u],\lambda,y),\qquad
 \partial_y\Phi([u],\lambda,y)=M([u],\lambda)\Phi([u],\lambda,y),
\label{eq4}
\end{equation}
where the matrices $M$ and $L$ are independent of the auxiliary variable $y$, i.e. $\partial_yL=\partial_yM=0,$
%
while the wavefunction $\Phi$ depends on $[u]$, $\lambda$ and the auxiliary variable $y$. Then the wavefunction $\Phi$ can be given in  the factored form
\begin{equation}
\Phi=\begin{pmatrix} w_1 & w_2 \\ \frac{{\rm d}}{{\rm d}x}w_1 & \frac{{\rm d}}{{\rm d}x}w_2 \end{pmatrix}\begin{pmatrix} e^{ay} & 0 \\ 0 & e^{-ay} \end{pmatrix}\in SL(2,\mathbb{C}),\hspace{5mm}a\in\mathbb{C},
\label{eq6}
\end{equation}
where $w_1$ and $w_2$ are two linearly independent particular solutions of the LSP (\ref{eq1}) which can be parametrized by a function $m(x)$ as follows
\begin{equation}
\begin{gathered}
 w_1=k_1m^{1/2}\exp\left(a\int\limits_{x_0}^x\frac{ds}{m}\right),\hspace{5mm}k_1\in\mathbb{C},\\
 w_2=k_1m^{1/2}\exp\left(a\int\limits_{x_0}^x\frac{ds}{m}\right)\left[k_2-\frac{1}{2ak_1^2}\exp\left(-2a\int_{x_0}^x\frac{ds}{m}\right)\right],\hspace{5mm}k_2\in\mathbb{C}.
\end{gathered}
\label{eq7}
\end{equation}
The general form of the $\mathfrak{sl}(2,\mathbb{C})$-valued matrix function $M$ is given by
\begin{equation*}
M=\begin{pmatrix} -\frac{1}{2}m_x & m \\ \frac{4a^2-m_x^2}{4m} & \frac{1}{2}m_x \end{pmatrix},
\label{eq8}
\end{equation*}
where the function $m$ satisfies the Gambier equation written as the linear third-order differential equation \cite[p. 27]{Gambier}
\begin{equation}
(\partial_x^3-4\alpha\partial_x-2\alpha_x)m=0.
\label{eq9}
\end{equation}
The contact transformation
$$
y_1:=\frac{m_x}{m},\qquad y_2:=\frac{m_{xx}}m
$$ 
maps G25 onto the projective Riccati equations
$$\left\{
\begin{aligned}
\frac{{\rm d}y_1}{{\rm d}x}&=y_2-y_1^2,\\
\frac{{\rm d}y_2}{{\rm d}x}&=4\alpha y_1+2\alpha_x-y_2y_1.\\
\end{aligned}\right.
$$
For arbitrary $x$-dependent coefficient $\alpha$, the previous Lie system is related to the Vessiot--Guldberg Lie algebras spanned by the vector fields
$$
Z_1:=(y_2-y_1^2)\partial_{y_1}-y_2y_1\partial_{y_2},\qquad Z_2:=y_1\partial_{y_2},\qquad Z_3:=\partial_{y_2}.
$$
In view of Table \ref{table3}, these vector fields are related to a Vessiot--Guldberg Lie algebra isomorphic to $\mathfrak{sl}(3)$.

Finally, it is worth noting that equation (\ref{eq9}) admits a first integral
\begin{equation}\label{eqqqq}
2mm_{xx}-m_x^2-4\alpha m^2+K=0,\qquad K\in \mathbb{R}.
\end{equation}
Equation (\ref{eqqqq}) can be written as a first-order system by adding a new variable $v=m_x$
$$
\left\{\begin{aligned}
\frac{{\rm d}m}{{\rm d}x}&=v,\\
\frac{{\rm d}v}{{\rm d}x}&=\frac{v^2}{2m}+2\alpha m.\\
\end{aligned}\right.
$$
This differential equations is related to an $x$-dependent vector field $X=2\alpha X_1+X_3$, where
$$
X_1:=v\frac{\partial}{\partial m}+\frac{v^2}{2m}\frac{\partial}{\partial v},\qquad X_2:=m\frac{\partial}{\partial m},\qquad X_3:=m\frac{\partial}{\partial v}
$$
have commutation relations
$$
[X_1,X_2]=X_1,\qquad [X_1,X_3]=-X_2,\qquad [X_2,X_3]=X_3.
$$
By using the Killing form, we obtain that the Lia algebra $V=\langle X_1,X_2,X_3\rangle$ is isomorphic to $\mathfrak{sl}(2)$.  In view of Table \ref{table3} and using that this Lie algebra spans a distribution of rank two, this Lie algebra must be diffeomorphic to one of the classes I$_4$, I$_5$, P$_2$. To determine exactly to which class $V$ is diffeomorphic to, we make use of of \cite[Theorem 4.4]{BHLS15}. The Casimir tensor field for this case reads
$$
\mathcal{R}=Y_1\otimes Y_3+Y_3\otimes Y_1+Y_2\otimes Y_2=m^2\partial_m\otimes \partial_m+v^2\partial_v\otimes\partial_v+vm(\partial_m\otimes \partial_v+\partial_v\otimes \partial_m).
$$
Hence, the determinant of the coefficients is zero and in view of the above-mentiond theorem $V$ is locally diffeomorphic to I$_5$. In view of the Table \ref{table2} and Table \ref{table3}, this is a Lie algebra of projective vector fields. We see from Table \ref{table5} that it is not related to a Lie algebra of conformal vector fields.

\section{Conclusions and Outlook}
The main objective of this study is to prove that the members of the Riccati hierarchy are equivalent to projective Riccati equations. This allows us to identify Riccati chain equations as the equations described by an $x$-dependent vector field taking values in the Lie algebra of projective vector fields of a flat Riemannian metric. The change of variables mapping the flat Riemannian metric into a diagonal form is the change of variables mapping Riccati chain equations into projective Riccati equations. As an application, we have derived superposition rules for all Riccati chain equations. 

We have studied the relations between the conformal and projective vector fields on $\mathbb{R}^2$ relative to different metrics. As additional results, we have proved several  propositions concerning the relations of inclusion between finite-dimensional Lie algebras of vector fields on the plane given in  \cite{BBHLS15}. Moreover, we found that the non-exhaustive relations of inclusion between Lie algebras given in that work are indeed all the relations that can be obtained in the case of projective and conformal vector fields. 

Finally, several applications of Riccati chain equations to the Sawada-Kotera and Kaup-Kupershmidt PDEs have been described. In addition, new relations between Gambier equations, Sturm-Liouville problems and the Riccati hierarchy have been established. 

In the future, we aim to show that most integrable PDEs can be studied through Lie systems. Additionally, we plan to use contact transformations to map Painlev\'e equations onto Lie systems. We also aim to study which differential equations can be mapped onto Lie systems through such transformations. Further exploration of relations between these two systems and their various properties are planned in our future work. This could increase the range of solvability of Lie systems by the technique described in this paper.

\section*{Acknowledgements}

The research of J. de Lucas was partially financed by the project MAESTRO under project number DEC-2012/06/A/ST1/00256. A.M. Grundland's work was supported by a research grant from the National Sciences and Engineering Research Council of Canada. J. de Lucas would also like to thank the Centre de Recherches Math\'ematiques (CRM) of the Universit\'e de Montr\'eal for its hospitality and attention during the research stay which gave rise to this work. Fruitful discussions on the present paper with P. Winternitz are also acknowleged.







\end{document}